%% file: f-minor-free-deletion.tex
\pdfoutput=1

\newif\ifllncs
\llncsfalse

\ifllncs
    \RequirePackage{amsmath}

    \documentclass[runningheads]{llncs}

    \usepackage[misc]{ifsym} % for the symbol indicating the corresponding author (\Letter)

    \title{A Turing Kernelization Dichotomy for Structural Parameterizations of \texorpdfstring{$\mathcal{F}$}{F}-Minor-Free Deletion}
    \titlerunning{Turing Kernelization Dichotomy for \FMDeletion}
    \author{Huib Donkers\,$^\text{\Letter}$\orcidID{0000-0002-2767-8140} \and
    Bart M.P. Jansen\thanks{Supported by NWO Gravitation grant ``Networks''.}\orcidID{0000-0001-8204-1268}}
    \authorrunning{H.T. Donkers \and B.M.P. Jansen}
    \institute{Eindhoven University of Technology,
    The Netherlands\\
    \email{\{h.t.donkers,b.m.p.jansen\}@tue.nl}}

    \usepackage{amssymb}
    \usepackage{xcolor}
    \usepackage{graphicx}
    \usepackage[tight]{subfigure}
    \usepackage[hyphens,spaces]{url}
    \usepackage[hidelinks]{hyperref}

    \let\origendproof=\endproof
    \def\endproof{\qed\origendproof}

    \let\origqed\qed
    \newcommand{\claimqed}{\hfill$\lrcorner$}
    \newenvironment{claimproof}[1][\proofname]{\begin{proof}\renewcommand{\qed}{\claimqed}}{\end{proof}\renewcommand{\qed}{\origqed}}

    \spnewtheorem{subclaim}{Claim}{\itshape}{\rmfamily}
    \spnewtheorem{observation}{Observation}{\bfseries}{\rmfamily}

\else

    \documentclass[a4paper,USenglish]{lipics-v2018-huib}
    \usepackage{microtype}%if unwanted, comment out or use option "draft"

    \title{A Turing Kernelization Dichotomy for Structural Parameterizations of \texorpdfstring{$\mathcal{F}$}{F}-Minor-Free Deletion}

    \titlerunning{Turing Kernelization Dichotomy for \FMDeletion}%optional, please use if title is longer than one line

    \author{Huib Donkers}{Eindhoven University of Technology, The Netherlands}{h.t.donkers@tue.nl}{0000-0002-2767-8140}{}%mandatory, please use full name; only 1 author per \author macro; first two parameters are mandatory, other parameters can be empty.

    \author{Bart M.\,P. Jansen}{Eindhoven University of Technology, The Netherlands}{b.m.p.jansen@tue.nl}{0000-0001-8204-1268}{Supported by NWO Gravitation grant 024.002.003 ``Networks''.}

    \authorrunning{H.\,T. Donkers and B.\,M.\,P. Jansen}%mandatory. First: Use abbreviated first/middle names. Second (only in severe cases): Use first author plus 'et al.'

    \Copyright{Huib Donkers and Bart M.\,P. Jansen}%mandatory, please use full first names. LIPIcs license is "CC-BY";  http://creativecommons.org/licenses/by/3.0/

    \subjclass{%
      ,
    }

    %mandatory

    \category{}%optional, e.g. invited paper

    \relatedversion{An extended abstract of this work appeared in the Proceedings of the 45th Workshop on Graph-Theoretical Concepts in Computer Science, WG 2019.}%optional, e.g. full version hosted on arXiv, HAL, or other respository/website

    \supplement{}%optional, e.g. related research data, source code, ... hosted on a repository like zenodo, figshare, GitHub, ...

    \funding{}%optional, to capture a funding statement, which applies to all authors. Please enter author specific funding statements as fifth argument of the \author macro.
 
    \EventEditors{John Q. Open and Joan R. Access}
    \EventNoEds{2}
    \EventLongTitle{42nd Conference on Very Important Topics (CVIT 2016)}
    \EventShortTitle{CVIT 2016}
    \EventAcronym{CVIT}
    \EventYear{2016}
    \EventDate{December 24--27, 2016}
    \EventLocation{Little Whinging, United Kingdom}
    \EventLogo{}
    \SeriesVolume{42}
    \ArticleNo{23}
    \nolinenumbers %uncomment to disable line numbering
    \theoremstyle{plain}

    \newtheorem{proposition}[theorem]{Proposition}
    \newtheorem{proposition*}{Proposition}
    \newtheorem{claim}[theorem]{Claim}
    \newtheorem{observation}[theorem]{Observation}

    \newcommand{\llheading}[1]{\subparagraph*}
    
    \newenvironment{subclaim}{\begin{claim}}{\end{claim}}
    
    \let\plainqed\qedsymbol
    \newcommand{\claimqed}{$\lrcorner$}
    \newenvironment{claimproof}[1][\proofname]{\begin{proof}\renewcommand{\qedsymbol}{\claimqed}}{\end{proof}\renewcommand{\qedsymbol}{\plainqed}}
    
    \hideLIPIcs
\fi

\usepackage{thmtools}
\usepackage{xpatch}
\makeatletter
\xpatchcmd{\thmt@restatable}% Edit \thmt@restatable
{\csname #2\@xa\endcsname\ifx\@nx#1\@nx\else[{#1}]\fi}% Replace this code
{\ifthmt@thisistheone\csname #2\@xa\endcsname\else\csname #2\@xa\endcsname\ifx\@nx#1\@nx\else[{#1}]\fi\fi}% with this code
{}{} % execute code for success/failure instances
\makeatother

\usepackage[inline]{enumitem}

\usepackage[ruled,vlined,linesnumbered]{algorithm2e}
\SetArgSty{textnormal}
\let\oldnl\nl% Store \nl in \oldnl
\newcommand{\nonl}{\renewcommand{\nl}{\let\nl\oldnl}}% Remove line number for one line

\mathchardef\hyphen="2D

\newcommand{\mx}{\uparrow}
\newcommand{\mn}{\downarrow}
\DeclareMathOperator{\tw}{tw}

\newcommand{\minorof}{\preceq}
\newcommand{\wminorof}{\precsim}
\DeclareMathOperator{\slb}{\lambda}
\DeclareMathOperator{\prune}{\hyphen prune}
\newcommand{\poly}[1]{\ensuremath{#1^{\Oh(1)}}}
\DeclareMathOperator{\odd}{odd}
\DeclareMathOperator{\isolated}{isol}
\DeclareMathOperator{\degree}{deg}
\DeclareMathOperator{\vc}{\textsc{vc}}
\DeclareMathOperator{\fvs}{\textsc{fvs}}

\newcommand{\Oh}{\ensuremath{\mathcal{O}}\xspace}
\newcommand{\F}{\ensuremath{\mathcal{F}}\xspace}
\newcommand{\Q}{\ensuremath{\mathcal{Q}}\xspace}
\newcommand{\mintw}{\ensuremath{\min\tw}\xspace}

\newcommand{\scVC}{\textsc{Vertex Cover}\xspace}
\newcommand{\scFVS}{\textsc{Feedback Vertex Set}\xspace}
\newcommand{\Deletion}[2]{\textsc{#1-#2-Free Deletion}\xspace}
\newcommand{\FMDeletion}{\Deletion{\F}{Minor}}
\newcommand{\FSDeletion}{\Deletion{\F}{Subgraph}}
\newcommand{\type}{\ensuremath{\mathit{type}}}
\newcommand{\FTDeletion}{\Deletion{\F}{\type}}

\newcommand{\NP}{\ensuremath{\mathsf{NP}}}
\newcommand{\NPhard}{\NP-hard\xspace}
\newcommand{\NPcomplete}{\NP-complete\xspace}
\newcommand{\coNPpoly}{\ensuremath{\mathsf{coNP/poly}}}
\newcommand{\MKtwo}{\ensuremath{\mathsf{MK[2]}}\xspace}
\newcommand{\MKtwohard}{\MKtwo-hard\xspace}
\newcommand{\MKtwocomplete}{\MKtwo-complete\xspace}

\newcommand{\notcontainment}{\ensuremath{\NP \not\subseteq \coNPpoly}\xspace}
\newcommand{\containment}{\ensuremath{\NP \subseteq \coNPpoly}\xspace}

\newcommand{\defparproblem}[4]{
\vspace{1mm}
\noindent\fbox{
\begin{minipage}{0.96\textwidth}
\textsc{#1} \\
{\bf{Input:}} #2 \\
{\bf{Parameter:}} #3 \\
{\bf{Question:}} #4
\end{minipage}
}
\vspace{1mm}
}

\begin{document}

\maketitle

\begin{abstract}
 For a fixed finite family of graphs \F, the \FMDeletion problem takes as input a graph $G$ and an integer $\ell$ and asks whether there exists a set $X \subseteq V(G)$ of size at most $\ell$ such that $G-X$ is \F-minor-free. For $\F=\{K_2\}$ and $\F=\{K_3\}$ this encodes \scVC and \scFVS respectively. When parameterized by the feedback vertex number of $G$ these two problems are known to admit a polynomial kernelization. Such a polynomial kernelization also exists for any \F containing a planar graph but no forests. 
 In this paper we show that \FMDeletion parameterized by the feedback vertex number is \MKtwohard for $\F = \{P_3\}$. This rules out the existence of a polynomial kernel assuming \notcontainment, and also gives evidence that the problem does not admit a polynomial Turing kernel. Our hardness result generalizes to any \F not containing a $P_3$-subgraph-free graph, using as parameter the vertex-deletion distance to treewidth $\mintw(\F)$, where $\mintw(\F)$ denotes the minimum treewidth of the graphs in \F. 
 For the other case, where \F contains a $P_3$-subgraph-free graph, we present a polynomial Turing kernelization. Our results extend to \FSDeletion.
\end{abstract}

\input{introduction.tex}

\input{prelimenaries.tex}

\input{lowerbound.tex}

\input{upperbound.tex}

\input{conclusion.tex}

\ifllncs
    % Handle underscores in doi neatly
    \def\doi#1{\href{https://doi.org/#1}{\nolinkurl{#1}}}
\fi
%\bibliography{f-minor-free-deletion}

\input{f-minor-free-deletion.bbl}
% \begin{appendix}
%  \input{appendix.tex}
% \end{appendix}

\end{document}

%% file: introduction.tex
\section{Introduction}

\ifllncs
  \subsubsection*{Background and motivation.}
\else
  \subparagraph*{Background and motivation}
\fi
Kernelization is a framework for the scientific investigation of provably effective preprocessing procedures for \NPhard problems. It uses the notion of a parameterized (decision) problem to capture meaningful performance guarantees for preprocessing. In a parameterized problem, every problem input~$x$ has an associated integer~$k$ called the parameter, which captures the difficulty of the input in some way. A \emph{kernelization} for a parameterized problem is a polynomial-time algorithm that transforms any parameterized instance~$(x,k)$ into an instance~$(x',k')$ with the same answer, such that~$|x'|$ and~$k'$ are both bounded by~$f(k)$ for some computable function~$f$. The function~$f$ is the \emph{size} of the kernel. Of particular interest are kernels of polynomial size. Determining which parameterized problems admit kernels of polynomial size has become a rich area of algorithmic research~\cite{Bodlaender09,FominLSZ19,LokshtanovMS12}.

A common approach in kernelization~\cite{AgrawalLMSZ17,FominLMS12,KratschW14,Iwata17} is to take the solution size as the parameter~$k$, with the aim of showing that large inputs that ask for a small solution can be efficiently reduced in size. However, this method does not give any nontrivial guarantees when the solution size is known to be proportional to the total size of the input. For that reason, there is an alternative line of research~\cite{BougeretS17,CyganLPPS14,FominJP14,GuoHN04,JansenB13,JansenK13,JansenP18,UhlmannW13} that focuses on parameterizations based on a measure of nontriviality of the instance (cf.~\cite{Niedermeier10}). One formal way to capture nontriviality of a graph problem, is to measure how many vertex-deletions are needed to reduce the input graph to a graph class in which the problem can be solved in polynomial time. Since many graph problems can be solved in polynomial time on trees and forests, the structural graph parameter \emph{feedback vertex number} (the minimum number of vertex deletions needed to make the graph acyclic, i.e., a forest) is a relevant measure of the distance of the input to a trivially solvable one.

Previous research has shown that for the \scVC problem, there is a polynomial kernel parameterized by the feedback vertex number~\cite{JansenB13}. This preprocessing algorithm guarantees that inputs which are large with respect to their feedback vertex number, can be efficiently reduced. The \scVC problem is the simplest in a family of so-called minor-free deletion problems. For a fixed finite family of graphs~\F, an input to \FMDeletion consists of a graph~$G$ and an integer~$\ell$. The question is whether there is a set~$S$ of at most~$\ell$ vertices in~$G$, such that the graph~$G-S$ obtained by removing these vertices does not contain any graph from~\F as a minor. Various classic graph optimization problems such as \scVC, \scFVS, and \textsc{Vertex Planarization} fit this framework by a suitable choice of~\F. The investigation of minor-free deletion problems has led to numerous advances in the study of kernelization and parameterized algorithmics~\cite{BasteST17,FominLMPS11,FominLMPS16,FominLMS12,GiannopoulouJLS15}. Motivated by the fact that \scVC and \scFVS, arguably the simplest \FMDeletion problems, admit polynomial kernels when parameterized by the feedback vertex number, we set out to resolve the following question: Do \emph{all} \FMDeletion problems admit a polynomial kernel when parameterized by the feedback vertex number?

\ifllncs
  \subsubsection*{Results.}
\else
  \subparagraph*{Results}
\fi
To our initial surprise, we prove that the answer to this question is \emph{no}. While the parameterization by feedback vertex number admits polynomial kernels for~$\F = \{K_2\}$~\cite{JansenB13}, for~$\F = \{K_3\}$~\cite{BodlaenderD10,Thomasse10,Iwata17}, and for any set~$\F$ containing a planar graph\footnote{If~$\F$ contains no forests, then any acyclic graph is~$\F$-minor free, implying the size of an optimal solution is at most the size of a feedback vertex set. Hence the kernelization for the solution-size parameterization yields a kernel of size bounded polynomially in the feedback vertex number.} but no forests~\cite{FominLMS12}, there are also cases that do not admit polynomial kernels (under the assumption that \notcontainment, which we tacitly assume throughout the informal discussion in this introduction). For example, the case of~$\F$ consisting of a single graph~$P_3$ that forms a path on three vertices does not admit a polynomial kernel. This lower bound for~$\F = \{P_3\}$ follows from a more general theorem that we state below.

Recall that a graph is a forest if and only if its treewidth is one~\cite{Bodlaender98}. Hence the feedback vertex number is exactly the minimum number of vertex deletions needed to obtain a graph of treewidth one. Let~$\tw(G)$ denote the treewidth of graph~$G$, and define~$\mintw(\F) := \min_{H \in \F} \tw(H)$. Our lower bound also holds for \FSDeletion, which is the related problem that asks whether there is a vertex set~$S$ of size at most~$k$ such that~$G-S$ contains no graph~$H \in \F$ as a \emph{subgraph}. We prove the following.

\begin{theorem} \label{thm:lowerbound}
Let~$\F$ be a finite set of graphs, such that each graph in~$\F$ has a connected component on at least three vertices. Then \FMDeletion and \FSDeletion do not admit polynomial kernels when parameterized by the vertex-deletion distance to a graph of treewidth~$\mintw(\F)$, unless \containment.
\end{theorem}

To see that Theorem~\ref{thm:lowerbound} implies the claimed lower bound for~$\F = \{P_3\}$, observe that whenever \F contains an acyclic graph with at least one edge we have~$\mintw(\F) = 1$ and therefore the vertex-deletion distance to treewidth~$\mintw(\F)$ equals the feedback vertex number. The theorem also generalizes earlier results of Cygan et al.~\cite[Theorem 13]{CyganLPPS14}, who investigated the problem of \emph{losing treewidth}. They proved that for each fixed~$1 \leq \eta < \rho$, the \textsc{$\eta$-Transversal} problem (delete at most~$\ell$ vertices to get a graph of treewidth at most~$\eta$) does not have a polynomial kernel when parameterized by the vertex-deletion distance to treewidth~$\rho$. Since the treewidth of a graph does not increase when taking minors, there is a finite set~$\F_\eta$ of forbidden minors (cf.~\cite{RobertsonS04}) that characterize the graphs of treewidth at most~$\eta$. As the members of the obstruction set for~$\eta \geq 1$ are easily seen to be connected, have treewidth~$\eta + 1$, and at least three vertices, the lower bound of Theorem~\ref{thm:lowerbound} encompasses the theorem of Cygan et al.~and generalizes it to arbitrary \FMDeletion problems.

Theorem~\ref{thm:lowerbound} is obtained through a polynomial-parameter transformation from the \textsc{cnf-sat} problem parameterized by the number of variables, for which a superpolynomial kernelization lower bound is known~\cite{DellM14,FortnowS11}. The main technical contribution in the hardness proof consists of the design of a gadget that acts as a clause checker. A certain budget of vertex deletions is available to break all \F-minors present in the gadget, and this is possible if and only if one of the neighboring vertices in a variable gadget is removed by the solution. This removal encodes that the variable is set in a way that satisfies the clause. The intricate part of the construction is to design the gadget knowing only that \F has a graph with a connected component of at least three vertices. Here we extensively rely on the fact that minimal minor models of biconnected graphs live in biconnected subgraphs, together with the fact the treewidth of a graph does not increase when attaching structures along cut vertices in a tree-like manner.

Using the framework of Hermelin et al.~\cite{HermelinKSWW15}, our polynomial-parameter transformation from \textsc{cnf-sat} parameterized by the number of variables to the structural parameterization of \FMDeletion, also rules out the existence of polynomial-size \emph{Turing} kernelizations under a certain hardness assumption. Turing kernelization~\cite{Fernau16} is a relaxation of the traditional form of kernelization. Intuitively, it investigates whether inputs~$(x,k)$ can be solved efficiently using the answers to subproblems of size~$f(k)$ which are provided by an oracle, which models an external computation cluster. Note that a parameterized problem that has a kernel of size~$\Oh(k^c)$ can be solved by a polynomial-time algorithm that first spends polynomial time to prepare a query of size~$\Oh(k^c)$, and then queries an oracle for its answer. Turing kernelization investigates if and how polynomial-time algorithms can solve \NPhard parameterized problems by querying an oracle for the answers to instances of size~$k^{\Oh(1)}$, potentially multiple times. Some problems that do \emph{not} admit polynomial kernelizations, do admit polynomial-size Turing kernelizations~\cite{Binkele-RaibleFFLSV12,Lokshtanov09,Jansen17,JansenPW17,Weller2013}.

Formally, a Turing kernelization of size~$f$ for a parameterized problem~$\Q$ is an algorithm that can query an oracle to obtain the answer to any instance of problem~$\Q$ of size and parameter bounded by~$f(k)$ in a single step, and using this power solves any instance~$(x,k)$ in time polynomial in~$|x|+k$. The reduction proving Theorem~\ref{thm:lowerbound} also proves the non-existence of polynomial-size Turing kernelizations, unless all parameterized problems in the complexity class \MKtwo defined by Hermelin et al.~\cite{HermelinKSWW15} have polynomial Turing kernels. (The \textsc{cnf-sat} problem with clauses of unbounded length, parameterized by the number of variables, is \MKtwocomplete~\cite[Thm.~1, cf.~Thm.~10]{HermelinKSWW15} and widely believed \emph{not} to admit polynomial-size Turing kernels.) 

Motivated by the general form of the lower bound statement in Theorem~\ref{thm:lowerbound}, we also investigate upper bounds and derive a complexity dichotomy. For any \F that does not meet the criterion of Theorem~\ref{thm:lowerbound}, we obtain a polynomial Turing kernel.

\begin{restatable}{theorem}{thmUpperboundStatement}
\label{thm:upperbound}
Let~$\F$ be a finite set of graphs, such that some~$H \in \F$ has no connected component of three or more vertices. Then \FMDeletion and \FSDeletion admit polynomial Turing kernels when parameterized by the vertex-deletion distance to a graph of treewidth~$\mintw(\F)$.
\end{restatable}

The main insight in the Turing kernelization is the following. If~$H \in \F$ has no connected component of three or more vertices, then~$H$ consists of disjoint edges and isolated vertices. If~$H$ only has isolated vertices, then~$\F$-minor-free deletion is polynomial-time solvable because the leftover graph has less than~$|V(H)| \in \Oh(1)$ vertices, for which we can search by brute force. Otherwise,~$H$ is a matching of size~$t \geq 1$ plus potentially some isolated vertices. The isolated vertices turn out only to make a difference if the solution $\F$-free graph has constant size. In the interesting case, we can focus on~$H \in \F$ being a matching of size~$t$. Then a graph that is~$\F$-minor-free does not admit a matching of size~$t$, and therefore has a vertex cover of size at most~$t$. Hence a solution to \FMDeletion can be extended to a vertex cover by including~$\Oh(1)$ additional vertices. Using the Tutte-Berge formula, we can make the relation between \FMDeletion and the vertex cover precise, and use it to reduce an instance of \FMDeletion parameterized by deletion distance to~$\mintw(\F)$, to the logical OR of a polynomial number of instances of \scVC parameterized by deletion distance to~$\mintw(\F)$. If~$\F$ has a graph with no component of size at least three, then~$\mintw(\F) = 1$, implying that the parameter is the feedback vertex set size. This allows us to use the polynomial kernel for \scVC parameterized by feedback vertex set on each generated instance. We query the resulting instances of size~$k^{\Oh(1)}$ to the oracle to find the answer.

\ifllncs
  \subsubsection*{Organization.}
\else
  \subparagraph*{Organization}
\fi
We present preliminaries on graphs and kernelization in Section~\ref{sec:prelims}. Section~\ref{sec:lowerbound} develops the lower bounds on (Turing) kernelization when all graphs in~$\F$ have a connected component with at least three vertices. In Section~\ref{sec:upperbound} we show that in all other cases, a polynomial Turing kernelization exists.

%% file: prelimenaries.tex
\section{Preliminaries} \label{sec:prelims}

All graphs we consider in this paper are simple, finite and undirected. We denote the vertex set and edge set of a graph $G$ by $V(G)$ and $E(G)$ respectively.
%For a vertex set $S \subseteq V(G)$ let $G[S]$ and $G-S$ denote the subgraph of $G$ induced by vertex set $S$ and $V(G)\setminus S$ respectively.
For a vertex set $S \subseteq V(G)$ let $G[S]$ be the subgraph of $G$ induced by $S$, and let $G - S$ denote the subgraph of $G$ induced by $V(G)\setminus S$.
For a vertex $v$ we use $G-v$ as shorthand for $G-\{v\}$.
For a non-negative integer $n$ we use $n \cdot G$ to denote the graph consisting of $n$ disjoint copies of $G$.
Let $N_G(S)$ and $N_G(v)$ denote the open neighborhood in $G$ of a vertex set $S$ and a vertex $v$ respectively.
Let $\degree_G(v)$ denote the degree of $v$ in $G$. The subscript may be omitted when $G$ is clear from the context. We use $\fvs(G)$ to denote the feedback vertex number of $G$.

A graph $H$ is a minor of graph $G$, denoted by $H \minorof G$, if $H$ can be obtained from $G$ by a series of edge contractions, edge deletions, and vertex deletions. An $H$-model in $G$ is a function $\varphi \colon V(H) \rightarrow 2^{V(G)}$ such that (i)~for every vertex $v \in V(H)$, the graph $G[\varphi(v)]$ is connected, (ii)~for every edge $\{u,v\} \in E(H)$ there exists an edge $\{u', v'\} \in E(G)$ with $u' \in \varphi(u)$ and $v' \in \varphi(v)$, and (iii)~for distinct~$u,v \in V(H)$ we have~$\varphi(v) \cap \varphi(u) = \emptyset$. The sets $\varphi(v)$ are called \emph{branch sets}. Clearly, $H \minorof G$ if and only if there is an $H$-model in $G$.
For any function $f \colon A \rightarrow B$ and set $A' \subseteq A$ we use $f(A')$ as a shorthand for $\bigcup_{a \in A'} f(a)$. Specifically in the case of a minor-model $\varphi$ and graph $G$, we use $\varphi(G)$ to denote $\bigcup_{v \in V(G)}\varphi(v)$. %$\varphi(V(G))$.
We say a graph $H$ is a \emph{component-wise minor} of a graph $G$, denoted as $H \wminorof G$, when every connected component of $H$ is a minor of $G$.

\begin{observation} \label{obs:minorisom}
 For graphs $G, H$, if $H \minorof G$ and $G \minorof H$ then $H$ and $G$ are isomorphic.
\end{observation}

\begin{observation} \label{obs:wminortw}
 For graphs $H, G$, if $H \wminorof G$ then $\tw(H) \leq \tw(G)$ since the treewidth of $H$ is the maximum treewidth of its connected components and each connected component of $H$ is a minor of $G$, so its treewidth it at most $\tw(H)$.
\end{observation}

\begin{definition}
 Let \F be a family of graphs and let $G \in \F$. For every relation $\trianglelefteq\; \in \{\minorof, \wminorof\}$ we define minimal and maximal elements as follows:
 \begin{itemize}
  \item $G$ is said to be \emph{$\trianglelefteq$-minimal} in \F when for all graphs $H \in \F$ we have $H \trianglelefteq G \Rightarrow G \trianglelefteq H$.
  \item $G$ is said to be \emph{$\trianglelefteq$-maximal} in \F when for all graphs $H \in \F$ we have $G \trianglelefteq H \Rightarrow H \trianglelefteq G$.
 \end{itemize}
\end{definition}

\begin{definition}
 We call a connected component $C$ of a graph $G$ a \emph{$\minorof$-maximal component} of $G$ when $C$ is \mbox{$\minorof$-maximal} in the set of graphs that form the connected components of $G$.
\end{definition}

For $\type \in \{\text{minor}, \text{subgraph}\}$ and a finite family of graphs \F, we define:

\defparproblem{\F-\type-Free Deletion}
{A graph $G$ and an integer $\ell$.}
{vertex-deletion distance to a graph of treewidth $\mintw(\F)$.}
{Is there a set~$X \subseteq V(G)$ of at most~$\ell$ vertices such that~$G-X$ does not contain any~$H \in \F$ as a \type?}

A vertex $v \in V(G)$ is a \emph{cut vertex} when its removal from $G$ increases the number of connected components.
A graph is called \emph{biconnected} when it is connected and contains no cut vertex. A \emph{biconnected component} of a graph $G$ is a maximal biconnected subgraph of $G$. 
%
%\begin{definition}
For any integer $\alpha$, a graph $G$ is called \emph{$\alpha$-robust} when $|V(G)| \geq \alpha$ and no vertex $v \in V(G)$ exists such that $G-v$ contains a connected component with less than $\alpha-1$ vertices.
%\end{definition}

\begin{proposition} \label{prop:uniquerobust}
 Any graph $G$ has a unique maximal $\alpha$-robust subgraph. Any $\alpha$-robust subgraph of~$G$ is a subgraph of the maximal $\alpha$-robust subgraph of~$G$.
\end{proposition}
\begin{proof}
 The proposition follows straightforwardly from the fact that if $G[A]$ and $G[B]$ are $\alpha$-robust, then so is $G[A \cup B]$. We now prove this fact.
 
 Consider two vertex sets $A, B \subseteq V(G)$, such that $G[A]$ and $G[B]$ are $\alpha$-robust. We show that $G[A \cup B]$ is $\alpha$-robust. Since $G[A]$ is $\alpha$-robust we have $|A| \geq \alpha$ so then $|A \cup B| \geq \alpha$. Suppose for contradiction that there exists a vertex $v \in A \cup B$ such that $G[A \cup B] - v$ contains a connected component of size smaller than $\alpha-1$. Let $C$ be the vertices of this connected component. We know $C$ contains vertices of at least one of $A$ and $B$. Assume w.l.o.g. $A \cap C \neq \emptyset$, then $G[A \cap C]$ is a connected component of size less than $\alpha-1$ in $G[A] - v$. If $v \in A$ this directly contradicts $\alpha$-robustness of $G[A]$, so assume $v \not\in A$. Now $G[A]$ contains a connected component with less than $\alpha-1$ vertices. Since $|C| < \alpha \leq |A|$ there exists a vertex $u \in A \setminus C$, so then $G[A] - u$ contains a connected component with less than $\alpha-1$ vertices, which contradicts $\alpha$-robustness of $G[A]$.
\end{proof}

%\begin{definition}
 For any graph $G$ and integer $\alpha$, let $\alpha\prune(G)$ denote the unique maximal $\alpha$-robust subgraph of $G$, which may be empty.
%\end{definition}
We define a \emph{leaf-block} of a graph $G$ as a biconnected component of $G$ that contains at most one cut vertex of $G$. The size of a leaf-block~$H$ is~$|V(H)|$. The size of the smallest leaf-block of a graph $G$ is denoted as~$\slb(G)$. Observe that~$G$ is $\alpha$-robust if and only if~$\slb(G) \geq \alpha$.

A \emph{polynomial-parameter transformation} from parameterized problem~$\mathcal{P}$ to parameterized problem~$\mathcal{Q}$ is a polynomial-time algorithm that, given an instance~$(x,k)$ of~$\mathcal{P}$, outputs an instance~$(x', k')$ of~$\mathcal{Q}$ such that all of the following are true:
\begin{enumerate}
	\item $(x,k) \in \mathcal{P} \Leftrightarrow (x', k') \in \mathcal{Q}$, 
	\item $k'$ is upper-bounded by a polynomial in~$k$.
\end{enumerate}

%% file: lowerbound.tex
%\section{When all graphs in \texorpdfstring{\F}{F} contain a \texorpdfstring{$P_3$}{P3}-subgraph}
\section{Lower bound}
\label{sec:lowerbound}

In this section we consider the case where all graphs in \F contain a connected component of at least three vertices and give a polynomial-parameter transformation from \textsc{cnf-sat} parameterized by the number of variables. In this construction we make use of the way biconnected components of graphs $G$ and $H$ restrict the options for an $H$-model to exist in $G$.

\begin{proposition} \label{prop:minmodelrobust}
 Let $H$ be an $\alpha$-robust graph and let $\varphi$ be a minimal $H$-model in a graph $G$, then $G[\varphi(H)]$ is $\alpha$-robust.
\end{proposition}
\begin{proof}
 Take an arbitrary vertex $v \in \varphi(H)$ and let $u \in V(H)$ be such that $v \in \varphi(u)$. Since $H-u$ does not have connected components smaller than $\alpha-1$, $G[\varphi(H)] - \varphi(u)$ cannot have connected components smaller than $\alpha-1$. Consider a spanning tree of $G[\varphi(u)]$. Each leaf of this spanning tree must be connected to a vertex in a different branch set, otherwise $\varphi$ is not minimal. We know every connected component in $G[\varphi(u)] - v$ contains at least one leaf of this spanning tree, hence every connected component of $G[\varphi(u)] - v$ is connected to $G[\varphi(H)] - \varphi(u)$. So $G[\varphi(H)] - v$ does not contain a connected component smaller than $\alpha-1$. Since $v$ was arbitrary, $G[\varphi(H)]$ is $\alpha$-robust.
\end{proof}

\begin{proposition} \label{prop:biconnectedminor}
 Let $\varphi$ be an $H$-model in $G$, and $B$ a biconnected component of~$H$. Then $G[\varphi(B)]$ contains a biconnected subgraph on at least $|B|$ vertices.
\end{proposition}
\begin{proof}
 Let $\varphi'$ a minimal $B$-model in $G$ such that $\varphi'(v) \subseteq \varphi(v)$ for all $v \in V(B)$. Hence $G[\varphi'(B)]$ is a subgraph of $G[\varphi(B)]$. It suffices to show that $G[\varphi'(B)]$ contains a biconnected component on at least $|V(B)|$ vertices. Since $B$ is biconnected, it is $|V(B)|$-robust so by Proposition~\ref{prop:minmodelrobust} we know $G[\varphi'(B)]$ is $|V(B)|$-robust. Hence $G[\varphi'(B)]$ contains a biconnected component on at least $|V(B)|$ vertices.
\end{proof}

\begin{observation} \label{obs:prune}
 For any graph $H$, which may be an empty graph on vertex set~$\emptyset$, and integers $\alpha \geq \beta$ we have $\alpha\prune(\beta\prune(H)) = \alpha\prune(H)$.
\end{observation}

\begin{proposition} \label{prop:prune}
 For graphs $H$ and $G$ we have $H \minorof G \Rightarrow \alpha\prune(H) \minorof \alpha\prune(G)$ for any integer $\alpha$.
\end{proposition}
\begin{proof}
 Clearly $\alpha\prune(H) \minorof H$ so since $H \minorof G$ we have $\alpha\prune(H) \minorof G$. For simplicity let $H' := \alpha\prune(H)$. Let $\varphi$ be a minimal $H'$-model in $G$ and observe $H' \minorof G[\varphi(H')]$. From Proposition~\ref{prop:minmodelrobust} we know $G[\varphi(H')]$ is $\alpha$-robust. Since $G[\varphi(H')]$ is an $\alpha$-robust subgraph of $G$ it is also a subgraph of $\alpha\prune(G)$ by Proposition~\ref{prop:uniquerobust}. Hence $\alpha\prune(H) \minorof G[\varphi(H')] \minorof \alpha\prune(G)$.
\end{proof}

\begin{proposition} \label{prop:pruning}
 For any $\trianglelefteq\; \in \{\minorof, \wminorof\}$, two integers $\alpha \geq \beta$, and graphs $H$ and $G$ we have that $H \trianglelefteq G \Rightarrow \alpha\prune(H)  \trianglelefteq \beta\prune(G)$.
\end{proposition}
\begin{proof}
 Suppose $H \minorof G$, then
 \begin{align*}
  \alpha\prune(H) &\minorof \alpha\prune(G) && \text{by Proposition~\ref{prop:prune}}\\
  &= \alpha\prune(\beta\prune(G)) && \text{by Observation~\ref{obs:prune}}\\
  &\minorof \beta\prune(G) .
 \end{align*}
 
 Alternatively, suppose $H \wminorof G$. Let $H'$ be a connected component of $\alpha\prune(H)$, then there exists a connected component $H''$ of $H$ such that $H' \minorof H''$. Since $H \wminorof G$ we have $H'' \minorof G$ so then $H' \minorof G$. As shown above, this implies $\alpha\prune(H') \minorof \beta\prune(G)$. Note that $\alpha\prune(H') = H'$ since $H'$ is a connected component of $\alpha\prune(H)$, hence $\alpha\prune(H) \wminorof \beta\prune(G)$.
 %Alternatively, when $H \wminorof G$ we know for all connected components $H'$ of $H$ that $H' \minorof G$. As shown before it follows that $\alpha\prune(H') \minorof \beta\prune(G)$, so all connected components of $\alpha\prune(H)$ are a minor of $\beta\prune(G)$, hence $\alpha\prune(H) \wminorof \beta\prune(G)$.
\end{proof}

 We proceed to construct a clause gadget to be used in the polynomial-parameter transformation from \textsc{cnf-sat}.

\begin{restatable}{lemma}{clausegadget} \label{lem:clausegadget}
 For any connected graph $H$ with at least three vertices there exists a polynomial-time algorithm that, given an integer $n \geq 1$, outputs a graph $G$ and a vertex set $S \subseteq V(G)$ of size~$n$ such that all of the following are true:
 \begin{enumerate}
  \item \label{clausegadget:1}$\tw(G) \leq \tw(H)$,
  \item \label{clausegadget:2}$G$ contains a packing of $3n-1$ vertex-disjoint $H$-subgraphs,
  \item \label{clausegadget:3}$G-S$ contains a packing of $3n-2$ vertex-disjoint $H$-subgraphs, and
  \item \label{clausegadget:4}$\forall v \in S$ there exists $X \subseteq V(G)$ of size $3n-1$ s.t. all of the following are true:
  \begin{enumerate}
   \item \label{clausegadget:4a}$v \in X$,
   \item \label{clausegadget:4b}$G-X$ is $H$-minor-free,
   \item \label{clausegadget:4c}$\slb(H)\prune(G-X) \wminorof H$, and
   \item \label{clausegadget:4d}for all connected components $G_c$ of $G-X$ that contain a vertex of $S$ we have $|V(G_c)| < \slb(H)$ and $G_c$ contains exactly one vertex of $S$.
  \end{enumerate}
 \end{enumerate}
\end{restatable}
\begin{proof}
%% CONSTRUCTION OF G
 Consider a subgraph $L$ of $H$ such that $L$ is a smallest leaf-block of $H$. Let $R$ be the graph obtained from $H$ by removing all vertices of $L$ that are not a cut vertex in $H$. Note that when $H$ is biconnected, $L=H$ and $R$ is an empty graph. We distinguish three distinct vertices $a,b,c$ in $H$. Vertices $c$ and $b$ are both part of $L$, where $c$ is the cut vertex (if there is one) and $b$ is any other vertex in $L$. Finally vertex $a$ is any vertex in $H$ that is not $c$ or $b$. See Fig.~\ref{fig:graphH}. In the construction of $G$ we will combine copies of $H$ such that $a$, $b$, and $c$ form cut vertices in $G$ and are part of two different $H$-subgraphs. Vertices $b$ and $c$ are chosen such that removing either one from a copy of $H$ in $G$ means no vertex from the $L$-subgraph of this copy of $H$ can be used in a minimal $H$-model in $G$.
 In the remainder of this proof we use $f_{K\rightarrow K'} \colon V(K) \rightarrow V(K')$ for isomorphic graphs $K$ and $K'$ to denote a fixed isomorphism.
 
 Take two copies of $H$, call them $H_1$ and $H_2$. Let $R_1$ and $L_1$ denote the subgraphs of $H_1$ related to $R$ and $L$, respectively, by the isomorphism between $H$ and $H_1$. Similarly let $R_2$ and $L_2$ denote the subgraphs of $H_2$. Take a copy of $L$ which we call $L_3$. Let $M$ be the graph obtained from the disjoint union of $H_1$, $H_2$, and $L_3$ by identifying the pair $f_{H \rightarrow H_1}(c)$ and $f_{H \rightarrow H_2}(b)$ into a single vertex $s$, and identifying the pair $f_{H \rightarrow H_2}(c)$ and $f_{L \rightarrow L_3}(c)$ into a single vertex $t$. We label $f_{H \rightarrow H_1}(a)$, $f_{H \rightarrow H_1}(b)$, and $f_{L \rightarrow L_3}(b)$ as $u$, $w$, and $v$ respectively.
 
 This construction is motivated by the fact that the graphs $M - \{v,s\}$, $M - \{u,t\}$, and $M - \{w,t\}$ are all $H$-minor-free, which we will exploit in the formal correctness argument later. We will connect copies of $M$ to each other via the vertices $u$, $v$, and $w$ so that, although two vertices need to be removed in every copy of $M$, one such vertex can always be in two copies of $M$ at the same time.
 
 \ifllncs
    \begin{figure}[bt]
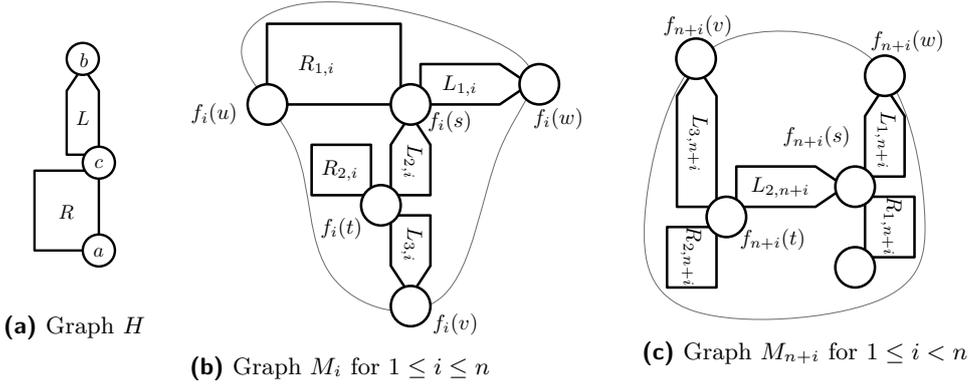

      \centering
      \subfigure[Graph $H$]{
	\includegraphics[scale=0.75,trim={-0.5cm -0.5cm -0.5cm 0},clip]{graphH-2.eps}
	\label{fig:graphH}
      }\hfill
      \subfigure[Graph $M_i$ for $1 \leq i \leq n$]{
	\includegraphics[scale=0.8]{graphMi-part1.eps}
	\label{fig:graphMi}
      }\hfill
      \subfigure[Graph $M_{n+i}$ for $1 \leq i < n$]{
	\includegraphics[scale=0.8]{graphMi-part2.eps}
	\label{fig:graphMni}
      }
      \caption{We show the situation where $a$ is contained in $R$. Note that $a$ can always be chosen such that it is contained in $R$ when $H$ is not biconnected. Note that the graphs in Fig.~\ref{fig:graphMi} and~\ref{fig:graphMni} are isomorphic but drawn differently.}
      \label{fig:graphHM}
    \end{figure}
 \else
    \begin{figure}[bt]
      \centering
      \begin{subfigure}{0.15\textwidth}
	\includegraphics[scale=0.75,trim={-0.5cm -0.5cm -0.5cm 0},clip]{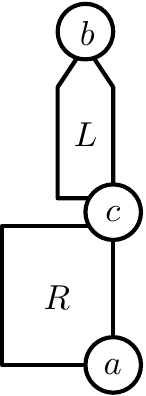}
	\caption{Graph $H$}
	\label{fig:graphH}
      \end{subfigure}
      \hfill
      \begin{subfigure}{0.40\textwidth}
	\includegraphics[scale=0.8]{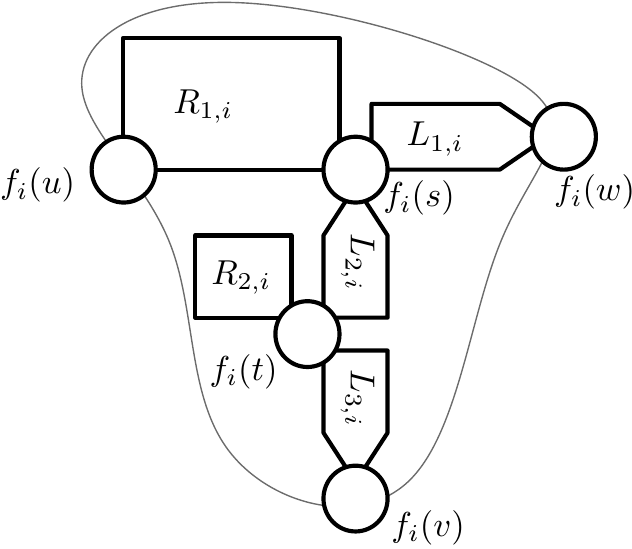}
	\caption{Graph $M_i$ for $1 \leq i \leq n$}
	\label{fig:graphMi}
      \end{subfigure}
      \hfill
      \begin{subfigure}{0.40\textwidth}
	\includegraphics[scale=0.8]{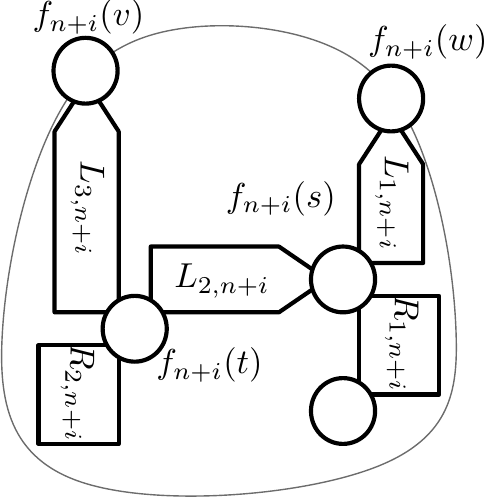}
	\caption{Graph $M_{n+i}$ for $1 \leq i < n$}
	\label{fig:graphMni}
      \end{subfigure}
      \caption{We show the situation where $a$ is contained in $R$. Note that $a$ can always be chosen such that it is contained in $R$ when $H$ is not biconnected. Note that the graphs in Fig.~\ref{fig:graphMi} and~\ref{fig:graphMni} are isomorphic but drawn differently.}
      \label{fig:graphHM}
    \end{figure}
 \fi

 Now take $2n-1$ copies of $M$, call them $M_1, \ldots, M_{2n-1}$. For readability we denote $f_{M \rightarrow M_i}$ as $f_i$ for all $1 \leq i \leq 2n-1$. For all $1 \leq i < n$ we identify $f_{i}(w)$ and $f_{n+i}(v)$, and we identify $f_{n+i}(w)$ and $f_{i+1}(u)$. Let this graph be $G$, and let $S$ be the set of vertices $f_i(v)$ for all $1 \leq i \leq n$. Let $H_{1,i}$, $H_{2,i}$, $R_{1,i}$, $R_{2,i}$, $L_{1,i}$, $L_{2,i}$, and $L_{3,i}$ denote the subgraphs in $M_i$ that correspond to the subgraphs $H_1$, $H_2$, $R_1$, $R_2$, $L_1$, $L_2$, and $L_3$ in $M$. See Fig.~\ref{fig:graphMi} and \ref{fig:graphMni}.
 
 This concludes the description of graph $G$ and set $S$. It is easily seen that these can be constructed in polynomial time. It remains to verify that all conditions of the lemma statement are met.
 
 %%PROOF
 \textbf{(\ref{clausegadget:1})} Since we connected copies of $L$ and $R$ in a treelike fashion along cut vertices, we did not introduce any new biconnected components. Since the treewidth of a graph is equal to the maximum treewidth over all its biconnected components we know that $\tw(G) \leq \max\{\tw(R), \tw(L)\} = \tw(H)$.
 
 \textbf{(\ref{clausegadget:2})} For each $1 \leq i \leq n$ we can distinguish two $H$-subgraphs in $M_i$, namely $H_{1,i}$ and $L_{3,i} \cup R_{2,i}$. This gives us $2n$ $H$-subgraphs in $G$. Note that since all $M_1, \dots, M_n$ are vertex-disjoint, these $2n$ $H$-subgraphs are also vertex-disjoint in $G$. For each $n < i \leq 2n-1$ we distinguish one $H$-subgraph, namely $H_{2,i}$. Note that since $H_{2,i}$ is vertex-disjoint from all $M_1, \dots, M_{i-1}, M_{i+1}, \dots, M_{2n-1}$ we have a total of $2n+n-1 = 3n-1$ vertex-disjoint $H$-subgraphs in $G$. This packing is shown in Fig.~\ref{fig:clause1}.

\textbf{(\ref{clausegadget:3})} Alternatively, for each $1 \leq i \leq n$ we can distinguish one $H$-subgraph in $M_i$, namely $H_{2,i}$. For each $n < i \leq 2n-1$ we distinguish two $H$-subgraphs in $M_i$, namely $H_{1,i}$ and $L_{3,i} \cup R_{2,i}$. Again these $H$-subgraphs are vertex-disjoint, and since they also do not contain any vertices of $S$, they form a packing of $n + 2(n-1) = 3n-2$ vertex-disjoint $H$-subgraphs in $G-S$. See Fig.~\ref{fig:clause2}.

\ifllncs
    \begin{figure}[bt]
      \centering
      \subfigure[A packing of $3n-1$ vertex-disjoint $H$-subgraphs in $G$]{
	\includegraphics[width=\textwidth]{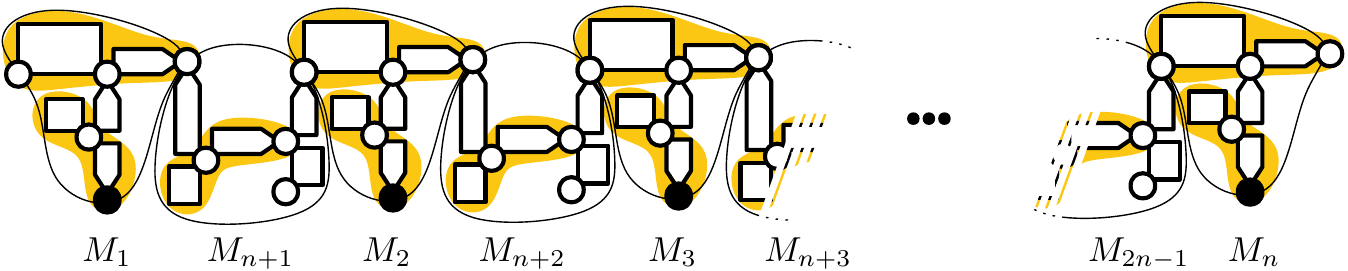}
	\label{fig:clause1}
      }
      \subfigure[A packing of $3n-2$ vertex-disjoint $H$-subgraphs in $G-S$]{
	\includegraphics[width=\textwidth]{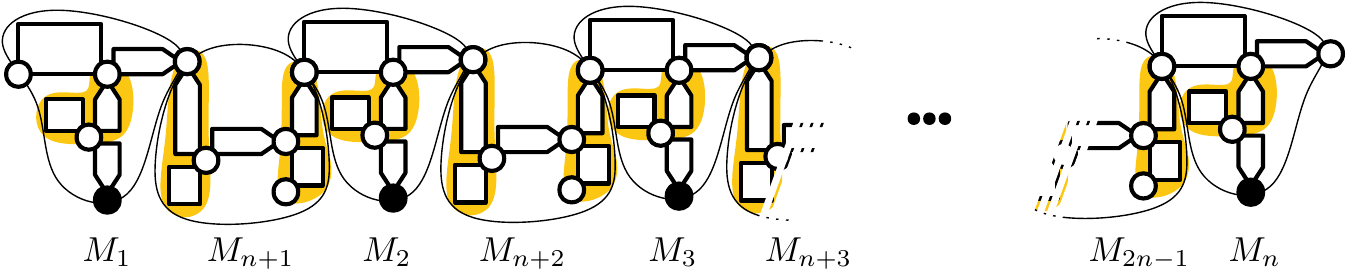}
	\label{fig:clause2}
      }
      \caption{Two packings of vertex-disjoint $H$-subgraphs in $G$ and $G-S$. Vertices in $S$ are marked black.}
    \end{figure}
\else
    \begin{figure}[bt]
      \centering
      \begin{subfigure}{\textwidth}
	\includegraphics[width=\textwidth]{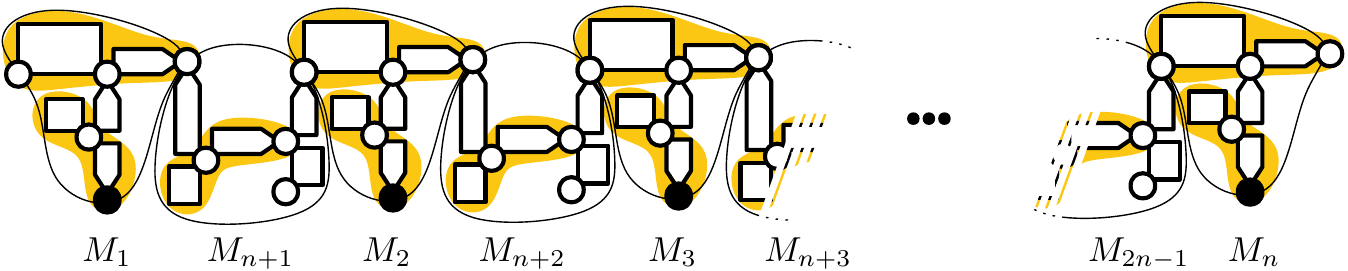}
	\caption{A packing of $3n-1$ vertex-disjoint $H$-subgraphs in $G$}
	\label{fig:clause1}
      \end{subfigure}
      
      \begin{subfigure}{\textwidth}
	\includegraphics[width=\textwidth]{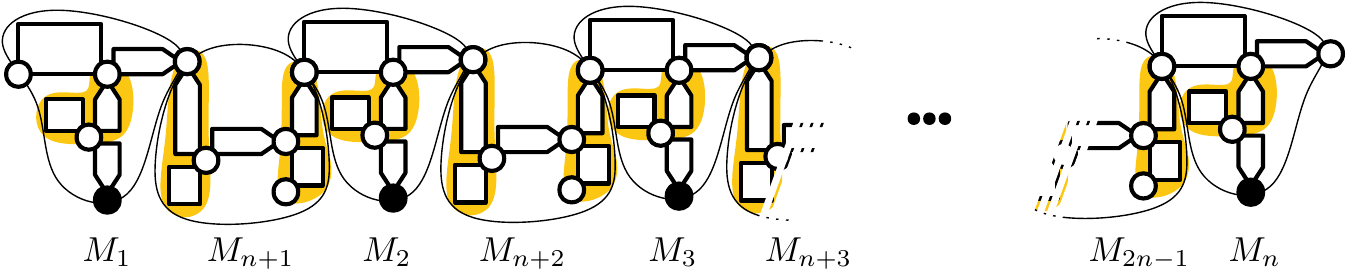}
	\caption{A packing of $3n-2$ vertex-disjoint $H$-subgraphs in $G-S$}
	\label{fig:clause2}
      \end{subfigure}
      \caption{Two packings of vertex-disjoint $H$-subgraphs in $G$ and $G-S$. Vertices in $S$ are marked black.}
    \end{figure}
\fi
 
 \textbf{(\ref{clausegadget:4})}
 Finally we prove that for all $v \in S$ there exists a set $Z \subseteq V(G)$ of size $3n-1$ such that the four conditions listed in condition~\ref{clausegadget:4} are true. For this purpose we first identify a family \Q of vertex sets such that any $H$-model in $G$ spans at least one vertex set in \Q. Let \Q be defined as follows: (see Fig.~\ref{fig:clause3})
 \begin{multline*}
  \Q =
  \{\{f_i(v),f_i(t)\} \mid 1 \leq i \leq 2n-1\}
  \cup \{\{f_i(t),f_i(s)\} \mid 1 \leq i \leq 2n-1\}\\
  \cup \{\{f_i(s),f_i(w)\} \mid n+1 \leq i \leq 2n-1\}
  \cup \{\{f_i(u),f_i(s),f_i(w)\} \mid 1 \leq i \leq n\}.\\
 \end{multline*}
 
 \begin{figure}[bt]
  \centering
  \includegraphics[width=\textwidth]{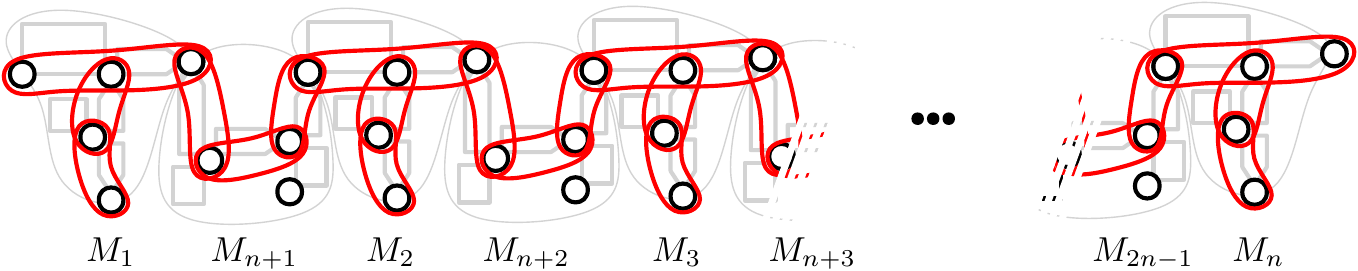}
  \caption{Vertex sets in \Q are encircled.}
  \label{fig:clause3}
 \end{figure}
 
 \begin{subclaim} \label{claim:clausegadget:QsubsetPhi}
  If $\varphi$ is an $H$-model in $G$, then $\varphi(H) \supseteq Q$ for some $Q \in \Q$.
 \end{subclaim}
 \begin{claimproof}
  Let $\varphi$ be an arbitrary $H$-model in $G$. We know from Proposition~\ref{prop:biconnectedminor} that $G[\varphi(L)]$ contains a biconnected subgraph on at least $|L|$ vertices. Let $B$ be such a biconnected subgraph in $G$. Subgraph $B$ must be fully contained in a biconnected component of $G$. Such a biconnected component must contain least $|B| \geq |L|$ vertices. We do a case distinction over all biconnected components in $G$ with size at least $|L|$, and prove that if $B$ is contained in them, then $\varphi(H) \supseteq Q$ for some $Q \in \mathcal{Q}$.
  \begin{itemize}
    \item $L_{2,i}$ for any $1 \leq i \leq 2n-1$: We know $|L_{2,i}| = |L|$ so $B = L$, hence $f_i(t),f_i(s) \in \varphi(L)$.
    \item $L_{3,i}$ for any $1 \leq i \leq 2n-1$: We know $|L_{3,i}| = |L|$ so $B = L$, hence $f_i(v),f_i(t) \in \varphi(L)$.
    \item $L_{1,i}$ for any $n+1 \leq i \leq 2n-1$: We know $|L_{1,i}| = |L|$ so $B = L$, hence $f_i(s),f_i(w) \in \varphi(L)$.
    \item $L_{1,i}$ for any $1 \leq i \leq n$: We know $|L_{1,i}| = |L|$ so $B = L$, hence $f_i(s),f_i(w) \in \varphi(L)$. If $f_i(t) \in \varphi(H)$ or $f_i(u) \in \varphi(H)$ then clearly $\mathcal{Q} \ni \{f_i(t),f_i(s)\} \subseteq \varphi(H)$ or $\mathcal{Q} \ni \{f_i(u),f_i(s),f_i(w)\} \subseteq \varphi(H)$. If $f_{n+i}(t) \in \varphi(H)$ then $\mathcal{Q} \ni \{f_{n+i}(v),f_{n+i}(t)\} \subseteq \varphi(H)$, since $f_i(w) = f_{n+i}(v)$. Suppose $\varphi(H)$ does not contain $f_i(t)$, $f_i(u)$, or $f_{n+i}(t)$, then $\varphi$ must be an $H$-model in the graph $G' := (H_{1,i} - f_i(u)) \cup (L_{2,i} - f_i(t)) \cup (L_{3,n+i} - f_i(t))$, so $H \minorof G'$. By Proposition~\ref{prop:prune} we know that $|L|\prune(H) \minorof |L|\prune(G')$. Clearly $|L|\prune(H) = H$. The graph $G'$ contains at least two leaf blocks that are smaller than $|L|$, namely $L_{2,i} - f_i(t)$ and $L_{3,n+i} - f_i(t)$, so $|L|\prune(G')$ is a subgraph of $H_{1,i} - f_i(u)$. But then $|V(|L|\prune(G'))| < |V(H)|$ so $|L|\prune(G')$ cannot contain an $H$-model. Contradiction.
    \item There can be biconnected components of size at least $|L|$ in $R_{2,i}$ for any $1 \leq i \leq 2n-1$. Suppose $f_i(t) \not\in \varphi(H)$, then $\varphi$ must be an $H$-model in the graph $R_{2,i} - f_i(t)$. Clearly this is not possible since $|V(R_{2,i} - f_i(t))| < |V(H)|$, so $f_i(t) \in \varphi(H)$. If $f_i(v) \in \varphi(H)$ or $f_i(s) \in \varphi(H)$ then $\mathcal{Q} \ni \{f_i(v),f_i(t)\} \subseteq \varphi(H)$ or $\mathcal{Q} \ni \{f_i(t),f_i(s)\} \subseteq \varphi(H)$. Suppose $\varphi(H)$ does not contain $f_i(v)$ or $f_i(s)$, then $\varphi$ must be an $H$-model in the graph $G' := R_{2,i} \cup (L_{3,i} - f_i(v)) \cup (L_{2,i} - f_i(s))$, so $H \minorof G'$. By Proposition~\ref{prop:prune} we know that $|L|\prune(H) \minorof |L|\prune(G')$, so $H \minorof |L|\prune(G') = R_{2,i}$. This is a contradiction since $R_{2,i}$ cannot contain $H$ as a minor.
    \item There can be biconnected components of size at least $|L|$ in $R_{1,i}$ for any $n+1 \leq i \leq 2n-1$. Suppose $f_i(s) \not\in \varphi(H)$, then $\varphi$ must be an $H$-model in the graph $R_{1,i} - f_i(s)$. Clearly this is not possible since $|V(R_{1,i} - f_i(s))| < |V(H)|$, so $f_i(s) \in \varphi(H)$. If $f_i(t) \in \varphi(H)$ or $f_i(w) \in \varphi(H)$ then $\mathcal{Q} \ni \{f_i(t),f_i(s)\} \subseteq \varphi(H)$ or $\mathcal{Q} \ni \{f_i(s),f_i(w)\} \subseteq \varphi(H)$. Suppose $\varphi(H)$ does not contain $f_i(t)$ or $f_i(w)$, then $\varphi$ must be an $H$-model in the graph $G' := R_{1,i} \cup (L_{2,i} - f_i(t)) \cup (L_{1,i} - f_i(w))$, so $H \minorof G'$. By Proposition~\ref{prop:prune} we know that $|L|\prune(H) \minorof |L|\prune(G')$, so $H \minorof |L|\prune(G') = R_{1,i}$. This is a contradiction since $R_{1,i}$ cannot contain $H$ as a minor.
    \item There can be biconnected components of size at least $|L|$ in $R_{1,i}$ for any $1 \leq i \leq n$. Suppose $f_{n+i-1}(s) \in \varphi(H)$ then $f_i(u) = f_{n+i-1}(w) \in \varphi(H)$ since any path in $G$ connecting $f_{n+i-1}(s)$ to any vertex in $R_{1,i}$ includes $f_i(u)$. So $\mathcal{Q} \ni \{f_{n+i-1}(s), f_{n+i-1}(w)\} \subseteq \varphi(H)$. Similarly if $f_i(t) \in \varphi(H)$ then $\mathcal{Q} \ni \{f_i(t),f_i(s)\} \subseteq \varphi(H)$ and if $f_{n+i}(t) \in \varphi(G)$ then $\mathcal{Q} \ni \{f_{n+i}(v),f_{n+i}(t)\} \subseteq \varphi(H)$. Suppose $\varphi(H)$ does not contain $f_{n+i-1}(s)$, $f_i(t)$ or $f_{n+i}(t)$, then $\varphi$ must be an $H$-model in $G' := H_{1,i} \cup (L_{1,n+i-1} - f_{n+i-1}(s)) \cup (L_{2,i} - f_i(t)) \cup (L_{3,n+i} - f_{n+i}(t))$. If $\mathcal{Q} \ni \{f_i(u),f_i(s),f_i(w)\} \subseteq \varphi(H)$, then the claim holds, so suppose $\{f_i(u),f_i(s),f_i(w)\} \not\subseteq \varphi(H)$, then for some $p \in \{f_i(u),f_i(s),f_i(w)\}$ we have that $\varphi$ is an $H$-model in $G' - p$. Therefore $H \minorof G'-p$ and by Proposition~\ref{prop:prune} we know $|L|\prune(H) \minorof |L|\prune(G'-p)$, so $H \minorof |L|\prune(G'-p) = |L|\prune(H_{1,i}-p)$. However $|L|\prune(H_{1,i}-p)$ has at most $|V(H_{1,i}-p)| = |V(H)|-1$ vertices, so it cannot contain an $H$-model. Contradiction.
  \end{itemize}
  This concludes the proof of Claim~\ref{claim:clausegadget:QsubsetPhi}.
 \end{claimproof}
 
 We now proceed to prove condition~\ref{clausegadget:4} of the lemma statement. Let $f_j(v) \in S$ be an arbitrary vertex in $S$, implying $1 \leq j \leq n$, and take
 \ifllncs
    \begin{equation*}
    X =
    \bigcup\limits_{1 \leq i < j} \{f_i(t), f_i(w), f_{i+n}(s)\}
      \cup
    \{f_j(v), f_j(s)\}
      \cup
    \bigcup\limits_{j < i \leq n} \{f_i(t), f_i(u), f_{i+n-1}(t)\} .
    \end{equation*}
\else
    \begin{equation*}
    X =
    \left(\bigcup\limits_{1 \leq i < j} \left\{f_i(t), f_i(w), f_{i+n}(s)\right\}\right)
      \cup
    \{f_j(v), f_j(s)\}
      \cup
    \left(\bigcup\limits_{j < i \leq n} \left\{f_i(t), f_i(u), f_{i+n-1}(t)\right\}\right) .
    \end{equation*}
\fi
 In Fig.~\ref{fig:clause5} the vertices in $Z$ are shown in graph $G$ as a red cross. Observe that $|X| = 3n-1$ and $f_j(v) \in X$. Furthermore $X$ contains at least one element from each set in $\Q$, hence $G-X$ so is $H$-minor-free by Claim~\ref{claim:clausegadget:QsubsetPhi}. This shows condition~\ref{clausegadget:4a} and \ref{clausegadget:4b} of the lemma statement holds. We proceed to show conditions~\ref{clausegadget:4c} and \ref{clausegadget:4d}.
 
 \begin{figure}[bt]
  \centering
  \includegraphics[width=\textwidth]{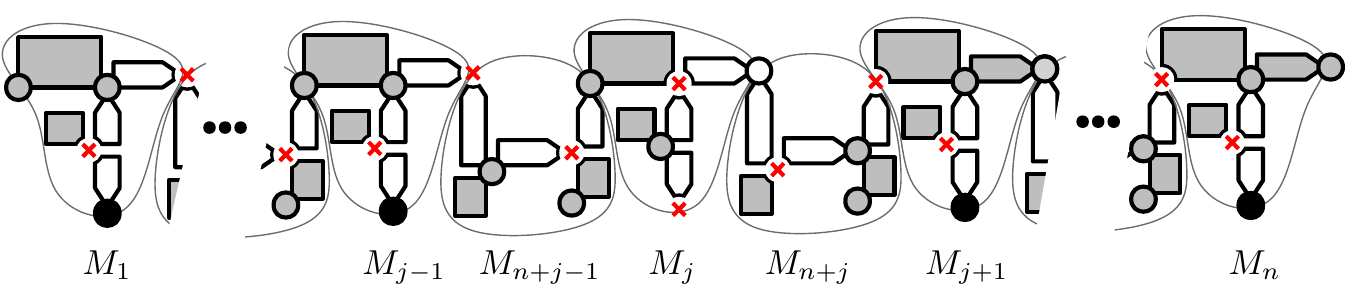}
  \caption{The graph $G-X$. Vertices in $X$ that are removed from the graph are marked by a cross. Vertices in $S$ are marked black. A supergraph of $\slb(H)\prune(G-X)$ is shown in gray. Note that when $|V(R)| = |V(L)|$, not all subgraphs and vertices marked gray are necessarily part of $\slb(H)\prune(G-X)$.}
  \label{fig:clause5}
 \end{figure}

\textbf{(\ref{clausegadget:4c})} Consider the graph $G' := \slb(H)\prune(G-X)$. Figure~\ref{fig:clause5} shows a supergraph of $G'$ in gray for the case that $a \in V(R)$. Every connected component in $G'$ must contain a biconnected component with at least $\slb(H)=|L|$ vertices. Consider all biconnected components in $G-X$ containing at least $|L|$ vertices, these can only be contained in the following subgraphs of $G$:
 \begin{itemize*}[label=]
  \item $R_{2,i}$ for any $1 \leq i \leq 2n-1$,
  \item $H_{1,i}$ for any $1 \leq i \leq n$, and
  \item $R_{1,i}$ for any $n+1 \leq i \leq 2n-1$.
 \end{itemize*}
 Note that any path from a vertex of one of these subgraphs to a vertex of another contains at least one vertex in $X$, hence any connected component in $G$ contains vertices of at most one of these subgraphs. Since all other biconnected components in $G-X$ have size less than $|L|$ we know that each connected component in $|L|\prune(G-X)$ is a subgraph of $R_{1,i}$, $R_{2,i}$ or $H_{1,i}$ for some $i$, hence $|L|\prune(G-X) \wminorof H$.
 
 \textbf{(\ref{clausegadget:4d})} Finally we show that all connected components in $G-X$ that contain a vertex of $S$ have size less than $|L|$. Since we have $f_j(v) \in X$, there is no connected component in $G-X$ containing $f_j(v)$. For all $i \neq j$ we have $f_i(t) \in X$ so the connected components in $G-X$ containing a vertex from $S$ are $L_{3,i} - f_i(t)$ for all $1 \leq i < j$ or $j < i \leq n$. These all have size $|L|-1$ and contain exactly one vertex of $S$.
\end{proof}

Using the clause gadget described in Lemma~\ref{lem:clausegadget} we give a polynomial-parameter transformation for the case where \F contains a single, connected graph $H$.

\begin{lemma} \label{lem:construction}
 For any connected graph $H$ with at least three vertices there exists a polynomial-time algorithm that, given a {\sc CNF}-formula $\Phi$ with $k$ variables, outputs a graph $G$ and an integer $\ell$ such that all of the following are true:
 \begin{enumerate}
  \item \label{construction:1}there is a set $S \subseteq V(G)$ of at most $2k$ vertices such that $\tw(G - S) \leq \tw(H)$,
  \item \label{construction:2}$G$ contains $\ell$ vertex-disjoint $H$-subgraphs,
  \item \label{construction:3}if $\Phi$ is not satisfiable then there does not exist a set $X \subseteq V(G)$ of size at most $\ell$ such that $G - X$ is \mbox{$H$-subgraph}-free,
  \item \label{construction:4}if $\Phi$ is satisfiable then there exists a set $X \subseteq V(G)$ of size at most $\ell$ such that $G - X$ is \mbox{$H$-minor}-free, $\slb(H)\prune(G - X) \wminorof H$, and $\tw(G-X) \leq \tw(H)$.
 \end{enumerate}
\end{lemma}

\begin{proof}
 Let $x_1, \dots, x_k$ denote the variables of $\Phi$, let $C_1, \dots, C_m$ denote the sets of literals in each clause of $\Phi$, and let $n$ denote the total number of occurrences of literals in~$\Phi$, i.e.~$n = \sum_{1 \leq j \leq m}|C_j|$.
 Let $H_1, \dots, H_k$ be copies of $H$. In each copy $H_i$ we arbitrarily label one vertex $v_{x_i}$ and another $v_{\neg x_i}$. Let $G_{var}$ be the graph obtained from the disjoint union of $H_1, \dots, H_k$. For each clause $C_j$ of $\Phi$ we create a graph called~$W_j$ and vertex set $S_j \subseteq V(W_j)$ by invoking Lemma~\ref{lem:clausegadget} with $H$ and $|C_j|$. Let $G$ be the graph obtained from the disjoint union of $W_1,\dots,W_m$ and $G_{var}$ where we identify the vertices in $S_j$ with the appropriate $v_{x_i}$ or $v_{\neg x_i}$ as follows: For each clause~$C_j$ let $s_1,\dots,s_{|C_j|}$ be the vertices in $S_j$ in some arbitrary order, and let $c_1,\dots,c_{|C_j|}$ be the literals in $C_j$, then we identify $s_i$ and $v_{c_i}$ for each $1 \leq i \leq |C_j|$. Finally let $\ell = k+3n-2m$ and $S = \bigcup_{1 \leq i \leq k} \{v_{x_i}, v_{\neg x_i}\}$. Note that $S_j \subseteq S$ for all $1 \leq j \leq m$. This concludes the description of $G$, $\ell$, and~$S$.
	
 It is easy to see they can be constructed in polynomial time. We proceed to show that all conditions in the lemma statement are met.
 
 \textbf{(\ref{construction:1})} Clearly $|S| = 2k$ and since every connected components in $G-S$ is a subgraph of $H_1, \dots, H_k$ or $W_1, \dots, W_m$, we can easily see that $\tw(G-S) \leq \tw(H)$.
 
 \textbf{(\ref{construction:2})} For all $1 \leq j \leq m$ we know from Lemma~\ref{lem:clausegadget} that $W_j - S$ contains a packing of $3|C_j|-2$ $H$-subgraphs. Since $W_j - S$ and $W_i - S$ are vertex-disjoint for $j \neq i$ we can combine these packings to obtain a packing in $G-S$ of $\sum_{1 \leq j \leq m} 3|C_j| - 2 = 3n-2m$ vertex-disjoint $H$-subgraphs. Note that this packing does not contain vertices from $H_1, \dots, H_k$, so we can add these to the packing and obtain a packing of $k+3n-2m = \ell$ vertex disjoint $H$-subgraphs in $G$.
 
 \textbf{(\ref{construction:3})} We now show that if $\Phi$ is not satisfiable, then there does not exist a set $X \subseteq V(G)$ of size at most $\ell$ such that $G - X$ is $H$-subgraph-free. Suppose there exists such a set $X$. Since there is a packing of $\ell$ vertex-disjoint $H$-subgraphs in $G$, we know that $X$ contains exactly one vertex from each $H$-subgraph in the packing. Since $v_{x_i}$ and $v_{\neg x_i}$ belong to the same subgraph, they cannot both be contained in $X$. Consider the variable assignment where $x_i$ is assigned \textit{true} if $v_{x_i} \in X$ or \textit{false} otherwise. Since we assumed $\Phi$ is not satisfiable, there is at least one clause in $\Phi$ that evaluates to \textit{false} with this variable assignment. Let $C_j$ denote such a clause. Since $C_j$ evaluates to \textit{false}, all of its literals must be \textit{false}, so for all variables $x_i$ that are not negated in $C_j$ we have $x_i = \textit{false}$ and therefore $v_{x_i} \not\in X$. For all negated variables $x_i$ in $C_j$ we know $x_i = \textit{true}$ meaning $v_{x_i} \in X$, so $v_{\neg x_i} \not\in X$. This means that $\emptyset = X \cap S_j = X \cap V(W_j) \cap V(G_{var})$, but since $G_{var}$ contains $k$ vertex-disjoint $H$-subgraphs we have $|X \cap V(G_{var})| \geq k$, so then $|X \cap (V(G_{var}) \setminus V(W_j))| \geq k$. For all $i$ there is a packing of $3|C_i|-2$ vertex-disjoint $H$-subgraphs in $W_i - S = W_i - V(G_{var})$, so in the graph $G - V(W_j)$ there are $k + \sum_{i \neq j}(3|C_i|-2)$ vertex-disjoint $H$-subgraphs. This means that $|X \cap (V(G) \setminus V(W_j))| \geq k + \sum_{i \neq j}(3|C_i|-2)$, and since $|X| = k + \sum_{1 \leq i \leq m}(3|C_i|-2)$ we know that $|X \cap V(W_j)| \leq 3|C_j|-2$. However $W_j$ contains $3|C_j|-1$ vertex-disjoint $H$-subgraphs, so $G-X$ cannot be $H$-subgraph-free. Contradiction.
 
 \textbf{(\ref{construction:4})} Finally we show that if $\Phi$ is satisfiable then there exists a set $X \subseteq V(G)$ of size at most $\ell$ such that $G - X$ is $H$-minor-free, $\slb(H)\prune(G - X) \wminorof H$, and $\tw(G-X) \leq \tw(H)$. Since $\Phi$ is satisfiable there exists a variable assignment such that each clause contains at least one literal that is true. Consider the set $X'$ consisting of all vertices $v_{x_i}$ when $x_i$ is \textit{true} and $v_{\neg x_i}$ when $x_i$ is \textit{false}. Since every clause contains one literal that is true, we know for each $1 \leq j \leq m$ that $W_j$ contains at least one vertex from $X'$. So for each $1 \leq j \leq m$ we have $X' \cap S_j \neq \emptyset$. Take and arbitrary vertex $v_j \in X' \cap S_j$ and let $X_j \subseteq V(W_j)$ be the vertex set containing $v_j$ obtained from condition~\ref{clausegadget:4} of Lemma~\ref{lem:clausegadget}. Let $X = X' \cup \bigcup_{1 \leq j \leq m} X_j$. For all $1 \leq j \leq m$ we know $|X' \cap X_j| \geq 1$ since $v_j \in X' \cap X_j$. So $|X| \leq |X'| + \sum_{1 \leq j \leq m} (3|C_j|-2) = k + 3n - 2m = \ell$.
%  Take an arbitrary vertex $v_j \in X' \cap S_j$. From Lemma~\ref{lem:clausegadget} we know there exists a set $X_j$ of size $3|C_j|-1$ such that all of the following are true:
%  \begin{itemize}
%   \item $v_j \in X_j$,
%   \item $W_j - X_j$ is $H$-minor-free,
%   \item $\slb(H)\prune(W_j-X_j) \wminorof H$ and
%   \item for all connected components $W_j'$ of $W_j-X_j$ that contain a vertex of $S_j$ we have $|V(W_j')| < \slb(H)$.
%  \end{itemize}
%  
% Let $X = X' \cup \bigcup_{1 \leq j \leq m} X_j$. For all $1 \leq j \leq m$ we know $|X' \cap X_j| \geq 1$ since $v_j \in X' \cap X_j$. So $|X| \leq |X'| + \sum_{1 \leq j \leq m} (3|C_j|-2) = k + 3n - 2m = \ell$.
 
 By condition~\ref{clausegadget:4b} we have that $W_j - X_j$ is $H$-minor-free for all $1 \leq j \leq m$, so clearly $W_j - X$ is also $H$-minor-free. Consider an arbitrary connected component $G'$ of $G-X$. If $G'$ is also a connected component of $W_j - X$ for some $1 \leq j \leq m$, then we have that $G'$ is $H$-minor-free, $\slb(H)\prune(G') \minorof H$ (by condition~\ref{clausegadget:4c}), and $\tw(G') \leq \tw(W_j) \leq \tw(H)$.
 If $G'$ is not a connected component of $W_j - X$ for any $1 \leq j \leq m$, then it contains a connected component of $H_i - X$ as a subgraph, for some $1 \leq i \leq k$. When $G'$ does not contain any vertices of $S$ we know that $G'$ must be a subgraph of $H_i$, so $G'$ is $H$-minor-free, $\slb(H)\prune(G') \minorof G' \minorof H$, and $\tw(G') \leq \tw(H_i) = \tw(H)$.
 
 Suppose on the other hand $G'$ does contain a vertex $v \in S$. No connected component of $W_j - X_j$ contains more than one vertex from $S$ and each connected component of $G_{var}$ contains exactly two vertices of $S$, one of which is in $X$. So $v$ is the only vertex in $G'$ that is contained in $S$. Moreover, since $S$ is the only overlap between the graphs $G_{var}$ and $W_j$ for all $1 \leq j \leq m$, we have that $v$ is a cut vertex in $G'$, such that for some $1 \leq i \leq k$, each biconnected component of $G'$ is a subgraph of $H_i - X$ or $W_j - X$ for any $1 \leq j \leq m$. So each of these biconnected components of $G'$ has treewidth at most $\tw(H)$, hence $\tw(G') \leq \tw(H)$. Also, each biconnected component in $G'$ that is a subgraph of $W_j - X = W_j - X_j$ for some $1 \leq j \leq m$ contains a vertex from $S$ and therefore has size at most $\slb(H)-1$ by condition~\ref{clausegadget:4d} on the choice of $X_j$. So we have that $\slb(H)\prune(G')$ is a subgraph of $H_i$, hence $\slb(H)\prune(G') \minorof G' \minorof H$. Additionally since $H_i$ contains at least one vertex that is not contained in $G'$ we have $H \not\minorof \slb(H)\prune(G')$. Because $H = \slb(H)\prune(H)$ we can conclude by Proposition~\ref{prop:pruning} that $G'$ is $H$-minor-free. Since $H$ is connected, and all connected components of $G - X$ are $H$-minor-free, $G-X$ must also be $H$-minor-free. We also know for all connected components $G'$ of $G - X$ that $\slb(H)\prune(G') \minorof H$, so $\slb(H)\prune(G-X) \wminorof H$. Finally since $\tw(G') \leq \tw(H)$ for each connected component $G'$ of $G - X$ we have that $\tw(G-X) \leq \tw(H)$.
\end{proof}

 The construction from Lemma~\ref{lem:construction} can directly be used to give a polynomial-parameter transformation from \textsc{cnf-sat} parameterized by the number of variables. Observe that if $G-X$ is \F-minor-free, then $G-X$ is also \F-subgraph-free. Similarly, if $G-X$ contains an $H$-subgraph for all $X \subseteq V(G)$ with $|X| \leq \ell$, then $G-X$ also contains an $H$-minor. Therefore, for any $\type \in \{\text{minor}, \text{subgraph}\}$ and \F consisting of one connected graph on at least three vertices, Lemma~\ref{lem:construction} gives a polynomial-parameter transformation from \textsc{cnf-sat} parameterized by the number of variables to \FTDeletion parameterized by deletion distance to $\mintw(\F)$.
 
 When \F contains multiple graphs, each containing a connected component of at least three vertices, it is possible to select a connected component $H$ of one of the graphs in \F such that the construction described in Lemma~\ref{lem:construction} forms the main ingredient for a polynomial-parameter transformation. We formally argue this in the following lemma.
 
 \begin{lemma} \label{lem:fconstruction}
 For a set \F of graphs, all with a connected component of at least $3$ vertices, and a {\sc CNF}-formula $\Phi$ with $k$ variables, we can create in polynomial time a graph $G$ and integer $\ell$ such that all of the following are true:
 \begin{enumerate}
  \item \label{fconstruction:twmod} there exists a set $S \subseteq V(G)$ of at most $\poly{k}$ vertices such that $\tw(G - S) \leq \mintw(\F)$,
%  \item \label{fconstruction:packing} $G$ contains $\ell$ vertex-disjoint $F$-subgraphs for some $F \in \F$ \red{(Not yet proven, possibly not even true in the current setting...!)},
  \item \label{fconstruction:notsat} if $\Phi$ is not satisfiable then there does not exist a set $X \subseteq V(G)$ of size at most $\ell$ such that $G - X$ is \F-subgraph-free, and
  \item \label{fconstruction:sat}if $\Phi$ is satisfiable then there exists a set $X \subseteq V(G)$ of size at most $\ell$ such that $G - X$ is \F-minor-free.
 \end{enumerate}
\end{lemma}

\begin{proof}
 Let $x_1,\ldots,x_k$ denote the variables of $\Phi$, let $C_1,\ldots,C_m$ denote the sets of literals in each clause of $\Phi$ and let $n$ denote the total number of literals in $\Phi$, i.e. $n = \sum_{1 \leq i \leq m}|C_i|$.
 
 Note that as a consequence of Observation~\ref{obs:wminortw}, there is a graph $F \in \F$ that is $\wminorof$-minimal with $\tw(F) = \mintw(\F)$. Let $\F_\mn \subseteq \F$ denote the set of all $\wminorof$-minimal graphs in $\F$ that have treewidth $\mintw(\F)$. Let $H_\mx$ denote a $\minorof$-maximal component of a graph in $\F_\mn$ such that no other $\minorof$-maximal component of a graph in $\F_\mn$ has a leaf-block smaller than $\slb(H_\mx)$. Let $H \in \F$ denote the graph that contains $H_\mx$ as $\minorof$-maximal component, so $H$ is $\wminorof$-minimal in \F and $\tw(H) = \mintw(\F)$.
 Let $c \geq 1$ denote the number of connected components in $H$ isomorphic to $H_\mx$ and let $Y$ denote the set of vertices contained in these connected components, so $G[Y] = c \cdot H_\mx$.
 
 Note that $H_\mx$ contains at least $3$ vertices since otherwise $H_\mx$ would be a minor of at least one connected component of $H$ containing at least $3$ vertices, which contradicts $H_\mx$ being a $\minorof$-maximal component of $H$.
 We use Lemma~\ref{lem:construction} to construct a graph $G'$ and integer $\ell'$ satisfying conditions~\ref{construction:1}--\ref{construction:4}. Let $S' \subseteq V(G')$ be the vertex set obtained from condition~\ref{construction:1}.
 
 Let $G_1 := (2c-1)\cdot G'$, and let the set $S$ be the union of all $2c-1$ corresponding copies of $S'$. Take $\ell = (2c-1)\cdot \ell'$ and let $G_2 := (\ell+1) \cdot (H-Y)$. We make the following claim about $G_2$:
 
 \begin{subclaim}\label{claim:G_2}
  $G_2$ has the following properties:
  \begin{enumerate*}[label=(\theenumi)]
    \item\label{G_2:wminor} $G_2 \wminorof H$,
    \item\label{G_2:tw} $\tw(G_2) \leq \tw(H)$,
    \item\label{G_2:Hfree} $G_2$ is $H$-minor-free, and
    \item\label{G_2:Ffree} $G_2$ is \F-minor-free.
  \end{enumerate*}
 \end{subclaim}
 \begin{claimproof}
  Property~\ref{G_2:wminor} follows directly from the construction and Property~\ref{G_2:tw} follows directly from Property~\ref{G_2:wminor}. To show Property~\ref{G_2:Hfree}, we show that $G_2$ is $H_\mx$-minor-free. Suppose for contradiction that $G_2$ contains $H_\mx$ as minor then, since $H_\mx$ is connected, there is a connected component~$H'$ of $G_2$ that contains $H_\mx$ as minor. $H'$ is also a connected component of $H$. Since $H_\mx$ is a $\minorof$-maximal component of $H$ and $H_\mx \minorof H'$ we know $H' \minorof H_\mx$, and it follows from Observation~\ref{obs:minorisom} that $H'$ is isomorphic to $H_\mx$. This is a contradiction since $G_2$ contains only connected components of $H$ that are not isomorphic to $H_\mx$.
  
  Having shown that $G_2$ is $H_\mx$-minor-free, Property~\ref{G_2:Ffree} is easily shown by contradiction. Suppose $G_2$ is not \F-minor-free, then there exists a graph $B \in \F$ such that $B \minorof G_2$. It follows from $G_2 \wminorof H$ that $B \wminorof H$ and since $H$ is $\wminorof$-minimal in \F we have that $H \wminorof B \minorof G_2$, but then $H_\mx \minorof G_2$. This is a contradiction since $G_2$ is $H_\mx$-minor-free.
  \end{claimproof}
 
 Consider the graph $G = G_2 \cup G_1$. It follows from condition~\ref{construction:1} of Lemma~\ref{lem:construction} that $\tw(G_1 - S) \leq \tw(H_\mx) \leq \tw(H)$ and since $\tw(G_2) \leq \tw(H)$ by Claim~\ref{claim:G_2}, we obtain $\tw(G-S) \leq \tw(H) = \mintw(\F)$. Note that $|S| = (2c-1)\cdot 2k \leq \Oh(k)$, hence condition~\ref{fconstruction:twmod} holds. We proceed to show the other conditions of the lemma statement.
 
 \textbf{(\ref{fconstruction:notsat})} Suppose $\Phi$ is not satisfiable, and take an arbitrary $X \subseteq V(G)$ of size at most $\ell$. We prove $G-X$ is not \F-subgraph-free by showing that $G - X$ contains an $H$-subgraph. First note that $G_2 - X$ contains at least one copy of $H-Y = H-c\cdot H_\mx$, so it remains to show that $G_1-X$ contains $c$ vertex-disjoint $H_\mx$-subgraphs. Recall that $G_1$ is the disjoint union of $2c-1$ copies of $G'$. Consider the subgraph $\hat{G}_1$ of $G_1$ consisting of the $G'$-subgraphs in $G_1$ that contain at most $\ell'$ vertices of $X$. Since $\Phi$ is not satisfiable, $G'$ leaves at least one $H_\mx$-subgraph when $\ell'$ or fewer vertices are removed, so each $G'$-subgraph in $\hat{G}_1$ leaves at least one $H_\mx$-subgraph in $G_1 - X$.
 When $\hat{G}_1$ contains at least $c$ vertex-disjoint $G'$-subgraphs, we know that there are at least $c$ vertex-disjoint $H_\mx$-subgraphs in $G_1 - X$, concluding the proof. Suppose instead that $\hat{G}_1$ contains less than $c$ vertex-disjoint $G'$-subgraphs. Let $x$ be the number of $G'$-subgraphs in $G_1 - V(\hat{G}_1)$. Since $G_1$ contains $2c-1$ vertex-disjoint $G'$-subgraphs we have $x \geq c$. Each of the $G'$-subgraphs in $G_1 - V(\hat{G}_1)$ contain at least $\ell'+1$ vertices of $X$, so $\hat{G}_1$ contains at most $\ell - x(\ell'+1)$ vertices of $X$. We also know $\hat{G}_1$ contains $\ell'((2c-1)-x)$ vertex-disjoint $H_\mx$-subgraphs since $G'$ contains $\ell'$ vertex-disjoint $H_\mx$-subgraphs and there are $(2c-1)-x$ vertex-disjoint $G'$-subgraphs in $\hat{G}_1$. We conclude that the number of vertex-disjoint $H_\mx$-subgraphs in $\hat{G}_1 - X$, and therefore also in $G_1 - X$, is at least
  \begin{align*}
   \ell'((2c-1)-x) - (\ell - x(\ell'+1))
   &= \ell'((2c-1)-x) - ((2c-1)\cdot \ell' - x\ell' - x)\\
   &= \ell'((2c-1)-x) - \ell'((2c-1)-x) + x \\
   &= x \geq c.
  \end{align*}
 This concludes the proof of condition~\ref{fconstruction:notsat}.
 
 \textbf{(\ref{fconstruction:sat})} When $\Phi$ is satisfiable we know that there exists a set $X' \subseteq V(G')$ of size at most $\ell'$ such that $G' - X'$ is $H_\mx$-minor-free and $\slb(H_\mx)\prune(G'-X') \wminorof H_\mx$. So then there exists a set $X \subseteq V(G_1)$ of size at most $(2c-1)\cdot \ell' = \ell$ such that $G_1 - X$ is $H_\mx$-minor-free and  $\slb(H_\mx)\prune(G_1-X) \wminorof H_\mx$. Since $G_2$ is also $H_\mx$-minor-free we know that $G - X$ is $H_\mx$-minor-free and therefore also $H$-minor-free. We now show that $G-X$ is also \F-minor-free.
 
 First observe the following:
 \begin{gather} \label{eqn:g1minorh}
  \slb(H_\mx)\prune(G_2 - X)
  \minorof G_2 - X
  \minorof G_2
  \wminorof H \text{, and} \\
   \label{eqn:g3minorh}
  \slb(H_\mx)\prune(G_1 - X)
  \wminorof H_\mx
  \minorof H .
 \end{gather}
We now deduce
 \begin{align*}
  \slb(H_\mx)\prune(G - X)
    &=\slb(H_\mx)\prune((G_2 - X) \cup (G_1 - X))
  \\&=\slb(H_\mx)\prune(G_2 - X) \cup \slb(H_\mx)\prune(G_1 - X))
  \\&\wminorof H \qquad \text{(by Equation~\ref{eqn:g1minorh} and \ref{eqn:g3minorh})}
 \end{align*}
 
 Suppose $G - X$ is not \F-minor-free, then for some $H' \in \F$ we have $H' \minorof G-X$. There must exist a graph $B \in \F$ such that $B$ is $\wminorof$-minimal in \F and $B \wminorof G-X$ since if $H'$ is $\wminorof$-minimal in \F then $H'$ forms such a graph $B$, and if on the other hand $H'$ is not $\wminorof$-minimal in \F then there exists a graph $H'' \in \F$ such that $H'' \wminorof H'$ and $H''$ is $\wminorof$-minimal in \F, meaning $H''$ forms such a graph $B$.
 
 Since $B \wminorof G-X$ we know by Observation~\ref{obs:wminortw} that $\tw(B) \leq \tw(G-X)$. Recall that $\tw(G-X) \leq \mintw(\F)$ so then $B \in \F_\mn$. Because of how we chose $H_\mx$, we know for all $\minorof$-maximal components $B_\mx$ of $B$ that $\slb(B_\mx) \geq \slb(H_\mx)$. Therefore
 \begin{align*}
  B
    &\wminorof \slb(H_\mx)\prune(B) &&\text{since } B = \slb(H_\mx)\prune(B)
  \\&\wminorof \slb(H_\mx)\prune(G-X) &&\text{by Proposition~\ref{prop:pruning} since } B \wminorof G-X
  \\&\wminorof H .
 \end{align*}
 
 Since $H$ is $\wminorof$-minimal in \F, it follows that $H \wminorof B$. By definition of $\wminorof$ we have $H_\mx \minorof B \wminorof G-X$. Since $H_\mx$ is connected we conclude $H_\mx \minorof G-X$. This is a contradiction since $G-X$ is $H_\mx$-minor-free.
 %\end{comment}
\end{proof}
 
 We conclude that a polynomial-parameter transformation exists for all $\type \in \{\text{minor},\allowbreak \text{subgraph}\}$ and \F containing only graphs with a connected component on at least three vertices. Together with the fact that \textsc{cnf-sat} is \MKtwohard and does not admit a polynomial kernel unless \containment (cf.~\cite[Lemma 9]{HermelinKSWW15}), this proves the following generalization of Theorem~\ref{thm:lowerbound}.

\begin{theorem}
 For $\type \in \{\text{minor}, \text{subgraph}\}$ and a set \F of graphs, all with a connected component of at least three vertices, \FTDeletion parameterized by vertex-deletion distance to a graph of treewidth $\mintw(\F)$ is \MKtwohard and does not admit a polynomial kernel unless \containment.
\end{theorem}

%% file: upperbound.tex
%\section{When a graph in \texorpdfstring{\F}{F} is \texorpdfstring{$P_3$}{P3}-subgraph-free}
\section{A polynomial Turing kernelization}
\label{sec:upperbound}

In this section we consider the case where \F contains a graph with no connected component of more than two vertices; or in short \F contains a $P_3$-subgraph-free graph. This graph consists of isolated vertices and disjoint edges. Let $\isolated(G)$ denote the set of isolated vertices in a graph $G$, i.e. $\isolated(G) = \{v \in V(G) \mid \degree(v) = 0\}$. We first show that the removal of all isolated vertices from all graphs in \F only changes the answer to \FMDeletion and \FSDeletion when the input is of constant size.

\begin{lemma} \label{lem:elimIsolated}
 For $\type \in \{\text{minor}, \text{subgraph}\}$ and any family of graphs \F containing a $P_3$-subgraph-free graph, let $\F' = \{F - \isolated(F) \mid F \in \F \}$. For any graph $G$, if $G$ is \F-\type-free but not $\F'$-\type-free, then $|V(G)| < \max\limits_{F \in \F} (|V(F)|+2|V(F)|^3)$.
\end{lemma}
\begin{proof}
 We first prove the lemma for $\type = \text{subgraph}$. Suppose $G$ is \F-subgraph-free but not $\F'$-subgraph-free. Now $G$ contains an $H'$-subgraph for some graph $H' \in \F'$. This subgraph consists of $|V(H')|$ vertices. Let $H \in \F$ be the graph for which $H' = H - \isolated(H)$. $G$ cannot contain $|\isolated(H)|$ vertices in addition to the vertices in the $H'$-subgraph because otherwise $G$ trivially contains an \F-subgraph. Hence $|V(G)| < |V(H')| + |\isolated(H)| = |V(H)| \leq \max\limits_{F \in \F} |V(F)|$.

 Next, we show the lemma holds for $\type = \text{minor}$. If some graph $G$ is \F-minor-free but not $\F'$-minor-free then for some graph $H \in \F$ we have $H' \minorof G$ but not $H \minorof G$ where $H' = H - \isolated(H)$. Let $\varphi$ be a minimal $H'$-model in $G$. The graph $G$ has less than $|V(\isolated(H))|$ vertices that are not in any branch set of $\varphi$, since otherwise an $H$-model could be constructed in $G$ by taking the branch sets of $\varphi$ and adding $|V(\isolated(H))|$ branch sets consisting of a single vertex.
 
 The number of vertices in $G$ that are contained in a branch set of $\varphi$ can also be limited. For an arbitrary vertex $v \in V(H')$ consider a spanning tree $T$ of $G[\varphi(v)]$. If $\varphi(v)$ contains multiple vertices then for each leaf $p$ of $T$, there must be a vertex $u \in N_{H'}(v)$ and $q \in N_G(p) \cap \varphi(u)$, such that $p$ is the only vertex from $\varphi(v)$ that is adjacent to $\varphi(u)$; otherwise, removing leaf $p$ from the branch set $\phi(v)$ would yield a smaller $H'$-model in $G$. Hence there can only be $\max\{1,\degree_{H'}(v)\}$ leaves in $T$.
 
 To give a bound on the size of each branch set consider a smallest graph $D \in \F'$ that is $P_3$-subgraph-free. Take $\ell = |V(D)|$ and note that $D \minorof P_\ell$. Since we know that $G$ is $\F'$-minor-free, $G$ must also be $P_\ell$-subgraph-free, therefore $T$ is also $P_\ell$-subgraph-free. Consider an arbitrary vertex $r$ in $T$. Since $T$ is a tree, there is exactly one path from $r$ to each leaf of $T$ and every vertex of $T$ lies on at least one path from $r$ to a leaf of $T$. Since there are no more than $\max\{1, \degree_{H'}(v)\}$ leaves in $T$ there are at most $\max\{1, \degree_{H'}(v)\}$ such paths, and all these paths contain less than $\ell$ vertices since $T$ is $P_\ell$-subgraph-free, hence in total $T$ contains less than $\degree_{H'}(v) \cdot \ell$ vertices. We can now give a bound on the total number of vertices in $G$ as follows: 
\begin{align*}
 |V(G)| 
 &< |\isolated(H)| + \sum\limits_{v \in H - \isolated(H)} |\varphi(v)| \\
 &\leq |\isolated(H)| + \sum\limits_{v \in H - \isolated(H)} (\degree_{H'}(v) \cdot \ell) \\
 &\leq |\isolated(H)| + 2 \cdot |E(H)| \cdot \ell \\
 &\leq |V(H)| + 2 \cdot |V(H)|^2 \cdot |V(D)| \\
 &\leq \max\limits_{F \in \F} (|V(F)|+2|V(F)|^3)
\end{align*}

This concludes the proof.
\end{proof}

After the removal of isolated vertices in \F to obtain $\F'$, we know that $\F'$ contains a graph consisting entirely of disjoint edges, i.e. this graph is isomorphic to $c \cdot P_2$ for some integer $c \geq 0$. If $c=0$ then \F-\type-free graphs have constant size and the problem is polynomial-time solvable. We proceed assuming $c \geq 1$. Let the matching number of a graph $G$, denoted as $\nu(G)$, be the size of a maximum matching in $G$. We make the following observation.
\begin{observation} \label{obs:matchingnumber}
 For all $c \geq 1$, a graph $G$ is $c\cdot P_2$-subgraph-free if and only if $\nu(G) \leq c-1$.
\end{observation}
We give a characterization of graphs with bounded matching number, based on an adaptation of the Tutte-Berge formula~\cite{Berge}. We use $\odd(G)$ to denote the number of connected components in $G$ that consist of an odd number of vertices.

\begin{lemma} \label{lem:partitioning-matchingnumber}
 For any graph $G$ and integer $m$ we have $\nu(G) \leq m$ if and only if $V(G)$ can be partitioned into three disjoint sets $U, R, S$ such that all of the following are true:
 \begin{itemize}
  \item all connected components in $G[R]$ have an odd size of at least $3$,
  \item $G[S]$ is independent,
  \item $N_G(S) \subseteq U$, and
  \item $|U| + \frac{1}{2} (|R| - \odd(G[R])) \leq m$.
 \end{itemize}
\end{lemma}
\begin{proof}
 Consider the Tutte-Berge formula \cite{Berge} (cf. \cite[Chapter 24]{Schrijver03}):
 \begin{equation*}
  \nu(G) = \tfrac{1}{2}\min\limits_{U \subseteq V(G)}(|V(G)| - \odd(G-U) + |U|) .
 \end{equation*}

 Suppose $\nu(G) \leq m$. It follows from the Tutte-Berge formula that there exists a $U_1 \subseteq V(G)$ such that $\frac{1}{2}(|V(G)| - \odd(G-U_1) + |U_1|) = \nu(G) \leq m$. From each connected component in $G-U_1$ on an even number of vertices, select a single non-cut vertex and add the selected vertices to a set $U_2$. Now take $U = U_1 \cup U_2$. Note that $G-U$ contains only odd sized connected components and $\odd(G-U) = \odd(G-U_1) + |U_2|$. Let $S$ be the set of isolated vertices in $G-U$ and let $R = V(G)\setminus(U \cup S)$. Observe that $U$, $R$, and $S$ satisfy the first three conditions in the lemma statement: $G[R]$ contains only connected components with an odd size of at least $3$, $G[S]$ is independent, and $N_G(S) \subseteq U$. Note that this implies that $\odd(G[R]) + |S| = \odd(G[R \cup S]) = \odd(G-U)$. The last requirement follows from the Tutte-Berge formula as follows:
 \begin{align*}
  |U| + \tfrac{1}{2} (|R| - \odd(G[R]))
  &= \tfrac{1}{2}(2|U| + |S| - |S| + |R| - \odd(G[R])\\
  &= \tfrac{1}{2}((|U| + |S| + |R|) - (\odd(G[R]) + |S|) + |U|)\\
  &= \tfrac{1}{2}(|V(G)| - \odd(G-U) + |U|) \\
  &= \tfrac{1}{2}(|V(G)| - (\odd(G-U_1) + |U_2|) + |U_1| + |U_2|) \\
  &= \tfrac{1}{2}(|V(G)| - \odd(G-U_1) + |U_1|) \\
  &= \nu(G) \leq m
 \end{align*}
 
 For the reverse direction of the proof, suppose $V(G)$ can be partitioned into disjoint sets $U,R,S$ as described in the lemma statement. A maximum matching in $G[R]$ has size at most $\frac{1}{2}(|R| - \odd(G[R])$ since at least one vertex in each odd component remains unmatched and every matching edge covers two vertices. Since $N_G(S) \subseteq U$ we know that $S$ is isolated in $G-U$, so $\nu(G-U) = \nu(G[R]) \leq \frac{1}{2}(|R| - \odd(G[R])$. Since a matching in $G$ is at most $|U|$ edges larger than a matching in $G-U$ we conclude $\nu(G) \leq |U| + \frac{1}{2}(|R| - \odd(G[R]) \leq m$.
\end{proof}

%Observe that for any partition~$U,R,S$ satisfying the first three conditions, we have~$|R| - \odd(G[R]) \geq \frac23|R|$ since each component of~$G[R]$ has at least three vertices. To satisfy the fourth condition therefore requires~$|U| + \frac12 (\frac23|R|) \leq m$, which will be a constant in our application. Since~$N_G(S) \subseteq U$, the lemma guarantees that~$N_G(R) \subseteq U$. 

Let us showcase how Lemma~\ref{lem:partitioning-matchingnumber} can be used to attack \FMDeletion when~$\F$ consists of a single graph~$c \cdot P_2$, so that the problem is to find a set~$X \subseteq G$ of size at most~$\ell$ such that~$G-X$ has matching number less than~$c$.

\begin{theorem} \label{thm:turingkernel:example}
For any constant~$c$, the \Deletion{$\{c \cdot P_2\}$}{Minor} problem parameterized by the size~$k$ of a feedback vertex set, can be solved in polynomial time using an oracle that answers \scVC instances with~$\Oh(k^3)$ vertices.
\end{theorem}
\begin{proof}
If an instance~$(G,\ell)$ admits a solution~$X$, then Lemma~\ref{lem:partitioning-matchingnumber} guarantees that~$V(G-X)$ can be partitioned into~$U,R,S$ satisfying the four conditions for~$m = c-1$. We try all relevant options for the sets~$U$ and~$R$ in the partition, of which there are only polynomially many since~$|U| + \frac13|R| \leq m \in \Oh(1)$.

For given sets~$U,R \subseteq V(G)$, we can decide whether there is a solution~$X$ of size at most~$\ell$ for which~$U,R$, and~$S := V(G) \setminus (U \cup R \cup X)$ form the partition witnessing that~$G-X$ has matching number at most~$m$, as follows. If some component of~$G[R]$ has even size, or less than three vertices, we reject outright. Similarly, if~$|U| + \frac12(|R| - \odd(G[R])) > m$, we reject. Now, if~$U$ and~$R$ were guessed correctly, then Lemma~\ref{lem:partitioning-matchingnumber} guarantees that the only neighbors of~$R$ in the graph~$G-X$ belong to~$U$. Hence we infer that all vertices of~$X' := N_G(R) \setminus U$ must belong to the solution~$X$. %Note that since~$S$ is an independent set in~$G-X$, the solution~$X$ contains an endpoint of every edge of~$G' := G - (U \cup R \cup X')$; hence the solution is a vertex cover of~$G'$. 
Note that since~$S$ is an independent set in~$G-X$, the solution~$X$ forms a vertex cover of~$G - (U \cup R)$, so that~$X'' := X \setminus X'$ is a vertex cover of~$G' := G - (U \cup R \cup X')$. 
On the other hand, for every vertex cover~$X''$ of~$G'$, the graph~$G - (X' \cup X'')$ will have matching number at most~$m$, as witnessed by the partition. Hence the problem of finding a minimum solution~$X$ whose corresponding graph~$G-X$ has~$U$ and~$R$ as two of the classes in its witness partition, reduces to finding a minimum vertex cover of the graph~$G'$. In terms of the decision problem, this means~$G$ has a solution of size at most~$\ell$ with~$U$ and~$R$ as witness partite sets, if and only if~$G'$ has a vertex cover of size at most~$\ell - |X'|$. Since~$\fvs(G') \leq \fvs(G)$, we can apply the known~\cite{JansenB13} kernel for \scVC parameterized by the feedback vertex number to reduce~$(G', \ell - |X'|)$ to an equivalent instance~with~$\Oh(\fvs(G)^3)$ vertices, which is queried to the oracle. If the oracle answers positively to any query, then~$(G,\ell)$ has answer \textsc{yes}; otherwise the answer is \textsc{no}.
\end{proof}

We remark that by using the polynomial-time reduction guaranteed by NP-completeness, the queries to the oracle can be posed as instances of the original \FMDeletion problem, rather than \scVC.
We now present our general (non-adaptive) Turing kernelization for the minor-free and subgraph-free deletion problems for all families~$\F$ containing a $P_3$-subgraph-free graph, combining three ingredients. Lemma~\ref{lem:elimIsolated} allows us to focus on families whose graphs have no isolated vertices. The guessing strategy of Theorem~\ref{thm:turingkernel:example} is the second ingredient. The final ingredient is required to deal with the fact that a solution subgraph~$G-X$ that is $c \cdot P_2$-minor-free for some~$c \cdot P_2 \in \F$, may still have one of the other graphs in~$\F$ as a forbidden minor. To cope with this issue, we show in Lemma~\ref{lem:shrinkS} that if~$G-X$ has no matching of size~$c$ (i.e., $G-X$ has a vertex cover of size at most $2c$), but does contain a minor model of some graph in~$\F$, then there is such a minor model of constant size. By employing a more expensive (but still polynomially bounded) guessing step, this allows us to complete the Turing kernelization. In the following lemmas $\vc(G)$ will denote the vertex cover number and $\Delta(G)$ will denote the maximum degree of $G$.
\begin{proposition}[{\cite[Proposition 1]{FominJP14}}] \label{prop:proposition1}
 If $G$ contains $H$ as a minor, then there is a subgraph $G^*$ of $G$ containing an $H$-minor such that $\Delta(G^*) \leq \Delta(H)$ and $|V(G^*)| \leq |V(H)| + \vc(G^*)\cdot(\Delta(H)+1)$.
\end{proposition}
\begin{lemma} \label{lem:shrinkS}
 For any $\type \in \{\text{minor}, \text{subgraph}\}$, let \F be a family of graphs, let $G$ be a graph with vertex cover $C$, and let $S = V(G-C)$. If $G$ contains an \F-\type, then there exists $S' \subseteq S$ such that $G[C \cup S']$ contains an \F-\type and $|S'| \leq \max_{H \in \F} |V(H)| + |C|\cdot(\Delta(H) + 1)$.
\end{lemma}
\begin{proof}
 Suppose $\type = \text{minor}$, then by Proposition~\ref{prop:proposition1} we know that if $G$ contains $H \in \F$ as a minor, then there is a subgraph $G^*$ of $G$ containing an $H$-minor such that $|V(G^*)| \leq |V(H)| + \vc(G^*)\cdot(\Delta(H) + 1)$. Take $S' = V(G^*) \cap S$, then $G[C \cup S'] = G[C \cup V(G^*)]$ contains an \F-minor and
 \begin{align*}
  |S'|
  &\leq |V(G^*)| \\
  &\leq |V(H)| + \vc(G^*)\cdot(\Delta(H) + 1) \\
  &\leq |V(H)| + |C|\cdot(\Delta(H) + 1) \\
  &\leq \max\limits_{H \in \F} |V(H)| + |C|\cdot(\Delta(H) + 1).
 \end{align*}
 
 On the other hand, when $\type = \text{subgraph}$ then $G$ contains an $H$-subgraph for some $H \in \F$, and trivially there exists a set $X \subseteq V(G)$ of $|V(H)|$ vertices such that $G[X]$ contains an $H$-subgraph. Take $S' = X - C$ and clearly $G[C \cup S']$ contains an $H$-minor.
\end{proof}

Armed with Lemma~\ref{lem:shrinkS} we now present the proof of the general Turing kernelization.

\thmUpperboundStatement*
\begin{proof}
 Fix some $\type \in \{\text{minor}, \text{subgraph}\}$. First, consider input instances $(G,\ell)$ for which $|V(G)| - \ell \leq \max_{F \in \F} (|V(F)| + 2|V(F)|^3)$. For these instances there exists a vertex set~$X$ of size at most $\ell$ such that $G-X$ is \F-\type-free if and only if there exists a vertex set~$Y$ of size at most $|V(G)| - \ell \leq \max_{F \in \F} (|V(F)| + 2|V(F)|^3)$ such that $G[Y]$ is \F-\type-free. Since there are only polynomially many such vertex sets~$Y$, and for each~$Y$ we can check in polynomial time whether $G[Y]$ contains an~\F-\type~\cite{RobertsonS95b}, we can apply brute force to solve the instance in polynomial time.
 
 So from now on we only consider instances $(G, \ell)$ for which $|V(G)| - \ell > \max_{F \in \F} (|V(F)| + 2|V(F)|^3)$. This means that for any vertex set~$X$ of size at most~$\ell$, the graph~$G-X$ contains more than $\max_{F \in \F} (|V(F)| + 2|V(F)|^3)$ vertices. Take $\F' = \{F - \isolated(F) \mid F \in \F\}$ and we obtain from Lemma~\ref{lem:elimIsolated} that if $G-X$ is $\F$-\type-free, it is also $\F'$-\type-free, and clearly if $G-X$ contains an \F-\type{} it is also contains an $\F'$-\type. Hence the \FTDeletion instance $(G,\ell)$ is equivalent to the \Deletion{$\F'$}{\type} instance $(G,\ell)$. Note that if $\F'$ contains an empty graph, the instance is trivially false since every graph contains the empty graph as a subgraph. In the rest of the algorithm we assume each graph in $\F'$ contains at least one edge.
 
 Since every graph in~\F contains an edge and at least one graph in~\F has no component of three vertices or more, we have $\mintw(\F) = 1$. Therefore the parameter, the deletion distance to treewidth~$\mintw(\F)$, is equal to $\fvs(G)$.
 
 To complete the Turing kernelization for \FTDeletion, it suffices to give a polynomial-time algorithm solving \Deletion{$\F'$}{\type} using an oracle that can solve \FMDeletion instances $(G', \ell')$ for which $|V(G')| \leq \poly{\fvs(G)}$ and $\fvs(G') \leq \poly{\fvs(G)}$. Note that the latter condition on $(G', \ell')$ is redundant since $\fvs(G') < |V(G')|$ for any graph $G'$.
 
 We will use the \FTDeletion oracle to solve \scVC instances on induced subgraphs $G_0$ of~$G$, for which $\fvs(G_0) \leq \fvs(G)$. This is done by a subroutine called \texttt{VCoracle}. A call to \texttt{VCoracle}$(G_0,\ell_0)$ decides whether~$G_0$ has a vertex cover of size at most~$\ell_0$, by the following process.
\begin{enumerate}
	\item Compute a 2-approximate feedback vertex set~$S_0$ on~$G_0$ in polynomial time, for example using the algorithm by Bafna et al.~\cite{BafnaBF99}.
	\item Apply the kernelization by Jansen and Bodlaender~\cite{JansenB13} for \scVC parameterized by feedback vertex set to the instance~$(G_0,\ell_0)$ and the approximate feedback vertex set~$S_0$. This takes polynomial time, and results in an instance~$(G_1,\ell_1)$ of \scVC on~$\Oh(|S_0|^3) \leq \Oh(\fvs(G_0)^3)$ vertices that is equivalent to~$(G_0,\ell_0)$.
	\item Since every graph in \F contains at least one edge, the \FTDeletion problem is \NPcomplete~\cite{LewisY80}. The \scVC problem is also known to be \NPcomplete, hence there exists a polynomial-time algorithm that transforms the \scVC instance~$(G_1,\ell_1)$ into an equivalent \FTDeletion instance~$(G_2, \ell_2)$. Since this algorithm runs in polynomial time and the size of its input is~$\poly{\fvs(G_0)}$, the number of vertices in~$G_2$ is upper-bounded by~$\poly{\fvs(G_0)}$.
	\item Query the instance~$(G_2, \ell_2)$ of size~$\poly{\fvs(G_0)}$ to the \FTDeletion oracle, and output the answer as the result of the call \texttt{VCoracle}$(G_0, \ell_0)$.
\end{enumerate}

Whenever~$G_0$ is an induced subgraph of~$G$, we have~$\fvs(G_0) \leq \fvs(G)$ so that the procedure above allows the Turing kernelization to implement an oracle for answering vertex cover queries on~$G_0$, using queries to the \FTDeletion oracle of size bounded polynomially in the parameter. 

Using this subroutine, the Turing kernelization algorithm is given in Algorithm~\ref{alg:kernel}. The high-level idea is as follows. Let~$M$ be the smallest graph in~$\F'$ that has no component of three or more vertices, or equivalently, which is~$P_3$-subgraph-free; then~$M$ consists of isolated edges. The Turing kernelization first guesses the sets~$U$ and~$R$ as per Lemma~\ref{lem:partitioning-matchingnumber} witnessing that the graph~$G-X$ obtained after removing the unknown solution~$X$, does not have a matching of~$|E(M)|$ edges (i.e.~that~$G-X$ does not contain~$M \in \F'$ as both a minor and a subgraph). Since Lemma~\ref{lem:partitioning-matchingnumber} guarantees that in the graph~$G-X$ we have~$N_{G-X}(R) \subseteq U$, it follows that~$N_G(R) \setminus U$ must belong to the unknown solution~$X$ if this guess was correct. The algorithm then considers the remaining vertices~$Q := V(G) \setminus (U \cup R \cup N_G(R))$ and classifies them into~$2^{|U|}$ types based on their adjacency to~$U$. An additional guessing step attempts to guess up to~$\alpha$ vertices of each type in~$G-X$, which will be part of the set~$S$ in the partition of Lemma~\ref{lem:partitioning-matchingnumber}, by taking them into the range of the function~$f$. The algorithm tests whether the graph~$G[f(2^U) \cup U \cup R]$ is $\F'$-\type-free. If not, then the guess was incorrect. If so, then for each type of which fewer than~$\alpha$ vertices were guessed to remain behind in~$G-X$, the algorithm collects the remaining vertices of that type in a set~$Q'$ to be added to the solution~$X$, and a \scVC instance is formulated on the remaining vertices of~$Q$. For types of which~$\alpha$ vertices remained behind, no vertices have to be added to~$Q'$ or the solution~$X$ in this step, because using Lemma~\ref{lem:shrinkS} it can be guaranteed that having more vertices of that type will not lead to an $\F'$-\emph{type}. The algorithm returns \emph{true} if the formulated instance of \scVC has a solution that yields a set of size at most~$\ell$ when combined with the vertices of~$N_G(R) \setminus U$ and~$Q'$.

 \begin{algorithm}[t]
  \caption{Solving \Deletion{$\F'$}{\type} instances using \texttt{VCoracle} with $\F'$ containing a $P_3$-subgraph-free graph and no empty graphs or graphs with isolated vertices.}
  \label{alg:kernel}
  \DontPrintSemicolon
  \SetKwFunction{oracle}{VCoracle}
  \SetKwFunction{true}{true}
  \SetKwFunction{false}{false}
  \SetKwInOut{Input}{input}
  \SetKwInOut{Output}{output}
  \SetKwComment{Comment}{$\triangleright$\ }{}
  
  \Input{A graph $G$ and an integer $\ell$}
  \Output{\true if there exists a set $X$ of size at most $\ell$ such that $G - X$ is $\F'$-\type-free, or \false otherwise.}
  
    Let $m = |E(M)|-1$ where $M$ is a smallest $P_3$-subgraph-free graph in $\F'$\;
    Let $\alpha = \max_{H \in \F'} |V(H)| + 3m(\Delta(H)+1)$\;
    \ForAll{$U \subseteq V(G)$ with $|U| \leq m$}{
      \ForAll{$R \subseteq V(G-U)$ such that \label{alg:line:choiceR}
	\;\Indp%\nonl% ugly...
      all connected components in $G[R]$ have an odd size of at least $3$
	\;%\nonl% ugly...
	{\normalfont and}
      $|U| + \frac{1}{2}(|R| - \odd(G[R])) \leq m$\;
      }{
	$Q := V(G)\setminus(U \cup R \cup N_G(R))$\; \label{alg:line:qdef}
	\ForAll{functions $f \colon 2^U \rightarrow 2^Q$ such that \label{alg:line:forallf}
	  \;\Indp%\nonl% ugly...
	  $G[f(2^U)]$ is independent
	  {\normalfont and}  \Comment*[r]{Recall that $f(2^U) = \bigcup_{Y\subseteq U}f(Y)$}
	  %\;%\nonl% ugly...
	  $G[f(2^U) \cup U \cup R]$ is $\F'$-\type-free \label{alg:line:choicef-minorfree}
	  {\normalfont and}\;%\nonl% ugly...
	  $\forall_{Y \subseteq U} |f(Y)| \leq \alpha$
	  {\normalfont and}\;%\nonl% ugly...
	  $\forall_{Y \subseteq U}\forall_{v \in f(Y)} N_G(v) \cap U = Y$
	  \;
	}{
	  $Q' := \{v \in Q \setminus f(2^U) \mid |f(N_G(v) \cap U)| < \alpha\}$\; \label{alg:line:qprime}
	  \If{\oracle{$G[Q] - Q'$, $\ell - |(N_G(R)\setminus U) \cup Q'|$}}{
	  \Return \true\; \label{alg:line:return2}
	  }
	}
      }
    }
  \Return \false\;
 \end{algorithm}

\ifllncs
    \begin{figure}[t]
      \centering
      \subfigure[]{
	\includegraphics[scale=0.7]{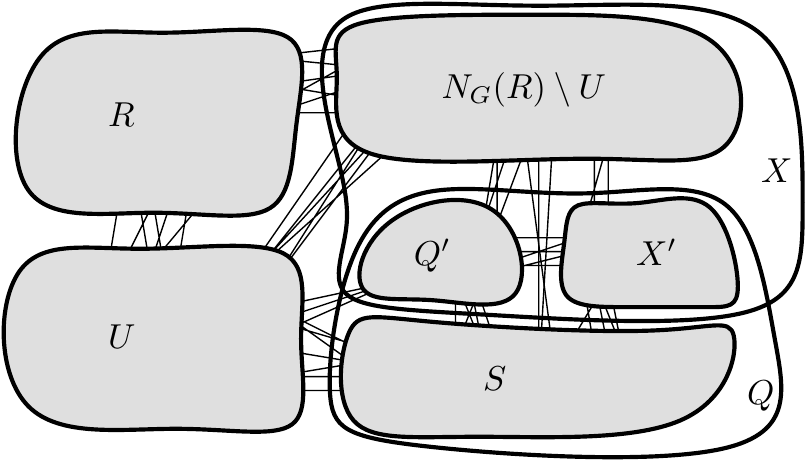}
	\label{fig:sets1}
      }\hfill
      \subfigure[]{
	\includegraphics[scale=0.7]{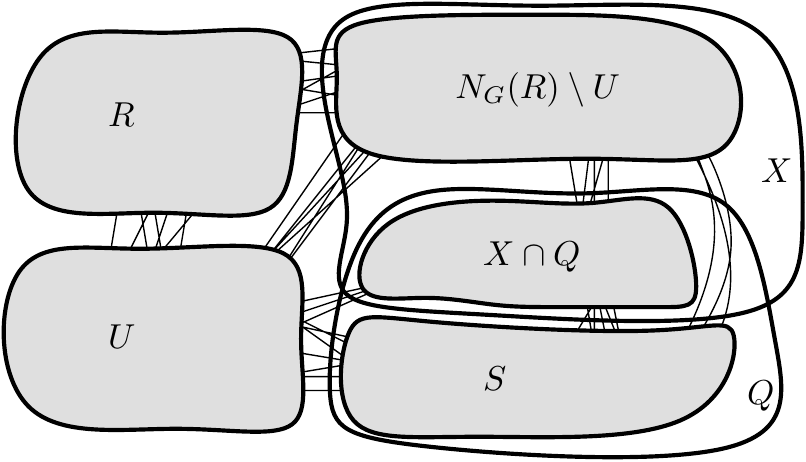}
	\label{fig:sets2}
      }
    \caption{We show two partitions of $G$. Figure~\ref{fig:sets1} shows a partition of $G$ given that Algorithm~\ref{alg:kernel} returns \textit{true}, while Fig.~\ref{fig:sets2} shows a partition of $G$ given that $G-X$ is \F-\type-free. Note that in both cases there can be no edges between $R$ and $Q$.}
    \end{figure}
\else
    \begin{figure}[t]
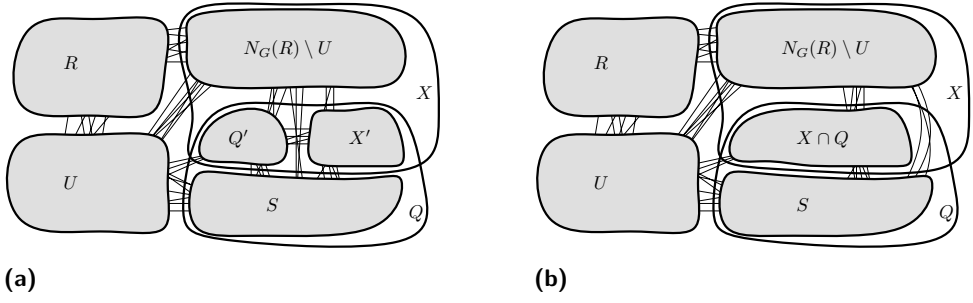

      \centering
      \begin{subfigure}{0.49\textwidth}
	\includegraphics[scale=0.7]{sets1-eps-converted-to.pdf}
	\caption{}
	\label{fig:sets1}
      \end{subfigure}
      \begin{subfigure}{0.49\textwidth}
	\includegraphics[scale=0.7]{sets2-eps-converted-to.pdf}
	\caption{}
	\label{fig:sets2}
      \end{subfigure}
    \caption{We show two partitions of $G$. Figure~\ref{fig:sets1} shows a partition of $G$ given that Algorithm~\ref{alg:kernel} returns \textit{true}, while Fig.~\ref{fig:sets2} shows a partition of $G$ given that $G-X$ is \F-\type-free. Note that in both cases there can be no edges between $R$ and $Q$.}
    \end{figure}
\fi
 
 \ifllncs
    \subsubsection*{Correctness.}
 \else
    \subparagraph*{Correctness}
 \fi
 When the algorithm returns \textit{true}, then consider the values of~$U$, $R$, $Q$, $f$, and $Q'$ at the time that \textit{true} is returned. There exists a vertex cover $X'$ of size at most~$\ell - |(N_G(R) \setminus U) \cup Q'|$ in~$G[Q] - Q'$. Let~$X = X' \cup (N_G(R)\setminus U) \cup Q'$, which has size at most~$\ell$. The set~$X$ is a vertex cover in~$G-(U \cup R)$, since $G-(U\cup R) - X = (G[Q] - Q') - X'$. Hence $S := V(G-(U \cup R)) \setminus X$ is an independent set, even in~$G$. The sets $U,R,S,X$ form a partition of $V(G)$. See Fig.~\ref{fig:sets1} for a visual representation of these sets. We will show that~$X$ is a solution to \Deletion{$\F'$}{\type} on~$G$.
 
 Consider an arbitrary vertex $v \in S$. Note that since $N_G(R) \subseteq U \cup X$ we have $S = V(G)\setminus(U \cup R \cup X) \subseteq V(G)\setminus(U \cup R \cup N_G(R)) = Q$, so $v \in Q$. By definition of $X$ we know $Q' \subseteq X$ so $v \not\in Q'$. Then by definition of $Q'$ on line~\ref{alg:line:qprime} we observe the following:
 \begin{observation} \label{obs:setS}
  For all $v \in S$ we have $v \in f(2^U)$ or $|f(N_G(v) \cap U)| \geq \alpha$.
 \end{observation}
 
 Assume for a contradiction that $G-X = G[U \cup R \cup S]$ contains an $\F'$-\type. Since $G[S]$ is independent, $U \cup R$ is a vertex cover in $G-X$, and by Lemma~\ref{lem:shrinkS} there exists a set $S' \subseteq S$ with $|S'| \leq \max\limits_{H \in \F'} |V(H)| + |U \cup R|\cdot(\Delta(H) + 1)$ such that $G[U \cup R \cup S']$ contains an $\F'$-\type. Note that $|R| - 3\odd(G[R]) \geq 0$ since every connected component in $G[R]$ contains at least $3$ vertices, so then
\allowdisplaybreaks
 \begin{align*}
  |S'|
  &\leq \max\limits_{H \in \F'} |V(H)| + |U \cup R|\cdot(\Delta(H) + 1)\\
  &\leq \max\limits_{H \in \F'} |V(H)| + (|U| + |R| + \frac{1}{2}(|R| - 3\odd(G[R])))\cdot(\Delta(H) + 1)\\
  &\leq \max\limits_{H \in \F'} |V(H)| + 3(|U| + \frac{1}{2}(|R| - \odd(G[R])))\cdot(\Delta(H) + 1)\\
  &\leq \max\limits_{H \in \F'} |V(H)| + 3m(\Delta(H) + 1)\\
  &= \alpha .
 \end{align*}
 \begin{subclaim} \label{claim:URSsubgraphURf}
  The graph $G[U \cup R \cup S']$ is isomorphic to a subgraph of $G[U \cup R \cup f(2^U)]$.
 \end{subclaim}
 \begin{claimproof}
  Observe that $S$ contains no neighbors of $R$, and since $G[S]$ is independent, we know for all $v \in S$ that $N_G(v) \subseteq U \cup X$ and therefore $N_{G-X}(v) = N_G(v) \cap U$. From Observation~\ref{obs:setS} it follows for all $v \in S'$ that $v \in f(2^U)$ or $|f(N_G(v) \cap U)| \geq \alpha$. In the latter case $v$ is a false twin of any vertex $u \in f(N_G(v) \cap U)$ in $G-X$ since by definition of $f$ we have $N_G(u) \cap U = N_G(v) \cap U$ for all vertices $u \in f(N_G(v) \cap U)$. We have $|S'| \leq |f(N_G(v) \cap U)|$ for all~$v \in S'$, so there exists a bijection that maps all vertices $v \in S'$ to a vertex in $u \in f(2^U)$ that is a false twin of $v$ in $G-X$. Any two false twins in $G-X$ are interchangeable in $G-X$, hence $G[U \cup R \cup S']$ is isomorphic to a subgraph of $G[U \cup R \cup f(2^U)]$.
 \end{claimproof}
 
 Since $f$ is chosen such that $G[U \cup R \cup f(2^U)]$ is $\F'$-\type-free on line~\ref{alg:line:choicef-minorfree}, Claim~\ref{claim:URSsubgraphURf} leads to a contradiction with the fact that~$G[U \cup R \cup S']$ contains an $\F'$-\type. We conclude that if the algorithm returns \textit{true} a set $X$ of size $\ell$ exists such that $G-X$ is $\F'$-\type-free.
 
 Next, we consider the reverse direction. We show that the algorithm returns \textit{true} when there exists a set $X$ of size at most $\ell$ such that $G-X$ is $\F'$-\type-free. Let $m = |E(M)| - 1$ where $M$ is the smallest $P_3$-subgraph-free graph in $\F'$, i.e. $M$ is isomorphic to $(m+1)\cdot P_2$ since no graph in $\F'$ contains isolated vertices. The graph $G-X$ is $\F'$-\type-free so it is also $(m+1) \cdot P_2$-subgraph-free, and by Observation~\ref{obs:matchingnumber} we know $\nu(G-X) \leq m$. Therefore by Lemma~\ref{lem:partitioning-matchingnumber} there exists a partition $U', R', S$ of $V(G-X)$ such that all of the following are true:
 \begin{itemize}
  \item all connected components in $(G-X)[R'] = G[R']$ have an odd size of at least $3$,
  \item $(G-X)[S] = G[S]$ is independent,
  \item $N_{G-X}(S) \subseteq U$ or equivalently $N_G(S) \subseteq U \cup X$, and
  \item $|U'| + \frac{1}{2} (|R'| - \odd(G[R'])) \leq m$.
 \end{itemize}
 Clearly $U'$ and $R'$ are such that there is an iteration in the algorithm where $U = U'$ and $R = R'$. Let $Q$ be the set as defined on line~\ref{alg:line:qdef} in this iteration, see Fig.~\ref{fig:sets2}. Let $g \colon 2^U \rightarrow 2^S$ be defined as $g(Y) = \{v \in S \mid Y = N_G(v) \cap U\}$ for all $Y \subseteq U$.
 We define a function $f' \colon 2^U \rightarrow 2^S$ that maps any $Y \subseteq U$ to an arbitrary subset of $g(Y)$ of size $\min\{|g(Y)|, \alpha\}$.
 We make the following observations:
 \begin{itemize}
  \item Since $N_G(S) \subseteq U \cup X$ we have $N_G(R) \cap S = \emptyset$ so $S = S \setminus N_G(R) = V(G) \setminus (U \cup R \cup X \cup N_G(R)) \subseteq V(G) \setminus (U \cup R \cup N_G(R)) = Q$, so $f' \colon 2^U \rightarrow 2^Q$.
  \item $G[S]$ is independent, so $G[f'(2^U)]$ is also independent because $f'(2^U) \subseteq S$.
  \item $G[U \cup R \cup f'(2^U)]$ is a subgraph of $G[U \cup R \cup S] = G-X$, and since $G-X$ is $\F'$-\type-free, $G[U \cup R \cup f'(2^U)]$ is also $\F'$-\type-free.
  \item Clearly $\forall_{Y \subseteq U} |f'(Y)| \leq \alpha$, and
  \item $\forall_{Y \subseteq U} \forall_{v \in f'(Y)} N_G(v) \cap U = Y$.
 \end{itemize}
 Hence $f'$ satisfies all conditions stated in line~\ref{alg:line:forallf} of the algorithm, so there is an iteration of the algorithm where $f = f'$. Let $Q'$ be the set as defined on line~\ref{alg:line:qprime} in this iteration. We now show that there exists a vertex cover of size at most $\ell - |(N_G(R)\setminus U) \cup Q'|$ in $G[Q] - Q'$.
 
 Since $G[S]$ is independent, $X$ is a vertex cover in $G[X \cup S] = G - (U \cup R)$. Then clearly $X \setminus (N_G(R) \setminus U)$ is a vertex cover in $G - (U \cup R \cup (N_G(R) \setminus U))$ and since $N_G(R) \setminus U \subseteq X$ we have $|X \setminus (N_G(R) \setminus U)| \leq \ell - |N_G(R)\setminus U|$. Similarly consider the set $A = (N_G(R)\setminus U) \cup Q'$. Clearly $X \setminus A$ is a vertex cover in $G - (U \cup R \cup A)$ and $|X \setminus A| \leq \ell - |A|$ if $A \subseteq X$. We will show that $A \subseteq X$. We know $N_G(R)\setminus U \subseteq X$ so it remains to be shown that $Q' \subseteq X$. Consider an arbitrary $v \in Q'$ and suppose $v \not\in X$. Since $Q' \subseteq Q$ we obtain from the definition of $Q$ that $v \not\in U$ and $v \not\in R$, so then $v \in S$. We also note from the definition of $Q'$ that $|f(N_G(v) \cap U)| < \alpha$. Since $f = f'$ we have $|f'(N_G(v) \cap U)| < \alpha$, and from the definition of $f'$ we know that if $|f'(Y)| < \alpha$ for some $Y \subseteq U$, then $f'(Y) = g(Y)$. By definition of $g$ we have $v \in g(N_G(v) \cap U)$, so then $v \in f(N_G(v) \cap U) \subseteq f(2^U)$. This is a contradiction since $v \not\in f(2^U)$ by definition of $Q'$.
 
 Now we have shown that $X\setminus A$ is a vertex cover of size at most $\ell - |A| = \ell - |(N_G(R) \setminus U) \cup Q'|$ in $G - (U \cup R \cup A) = G[Q] - Q'$, hence the \texttt{VCoracle} should report that a vertex cover exists on line $8$.
 
 \ifllncs
    \subsubsection*{Running time and query size.}
 \else
    \subparagraph*{Running time and query size}
 \fi
 Since \F is fixed~$m$ and $\alpha$ are constants, the sets $U$ and $R$ are of constant size so there are only polynomially many possibilities for $U$ and $R$. The function $f$ maps subsets of~$U$ to sets of a maximum size of~$\alpha$, so there are only polynomially many possible functions over which to iterate. For each possibility we run \texttt{VCoracle} which takes polynomial time.
 
 The \texttt{VCoracle} subroutine is invoked on induced subgraphs of~$G$ which therefore have a feedback vertex number of at most~$\fvs(G)$. Hence after computing the vertex cover kernel we invoke the oracle for vertex cover instances with~$\Oh(\fvs(G)^3)$ vertices.
\end{proof}

%% file: conclusion.tex
\section{Conclusion}

Earlier work~\cite{BodlaenderD10,Iwata17,JansenB13,Thomasse10} has shown that several \FMDeletion problems admit polynomial kernelizations when parameterized by the feedback vertex number. In this paper we showed that when \F contains a forest and each graph in~$\F$ has a connected component of at least three vertices, the \FMDeletion problem does \emph{not} admit such a polynomial kernel unless \containment. This lower bound generalizes to any \F where each graph has a connected component of at least three vertices, when we consider the vertex-deletion distance to treewidth $\mintw(\F)$ as parameter.

For all other choices of \F we showed that a polynomial Turing kernelization exists for \FMDeletion parameterized by the feedback vertex number. The size of the \scVC  queries generated by the Turing kernelization does not depend on \F: the Turing kernelization can be shown to be \emph{uniformly polynomial} (cf.~\cite{GiannopoulouJLS17}). However, it remains unknown whether the \emph{running time} can be made uniformly polynomial, and whether the Turing kernelization can be improved to a traditional kernelization.

Our results leave open the possibility that all \FMDeletion problems admit a polynomial kernel when parameterized by the vertex-deletion distance to a \emph{linear forest}, i.e.~a collection of paths. Resolving this question may be an interesting direction for future work.

%% file: f-minor-free-deletion.bbl
\begin{thebibliography}{10}

\bibitem{AgrawalLMSZ17}
Akanksha Agrawal, Daniel Lokshtanov, Pranabendu Misra, Saket Saurabh, and
  Meirav Zehavi.
\newblock Feedback vertex set inspired kernel for chordal vertex deletion.
\newblock In {\em Proc. 28th SODA}, pages 1383--1398. {SIAM}, 2017.
\newblock \href {http://dx.doi.org/10.1137/1.9781611974782.90}
  {\path{doi:10.1137/1.9781611974782.90}}.

\bibitem{BafnaBF99}
Vineet Bafna, Piotr Berman, and Toshihiro Fujito.
\newblock A 2-approximation algorithm for the undirected feedback vertex set
  problem.
\newblock {\em SIAM Journal on Discrete Mathematics}, 12(3):289--297, 1999.
\newblock \href {http://dx.doi.org/10.1137/S0895480196305124}
  {\path{doi:10.1137/S0895480196305124}}.

\bibitem{BasteST17}
Julien Baste, Ignasi Sau, and Dimitrios~M. Thilikos.
\newblock Optimal algorithms for hitting (topological) minors on graphs of
  bounded treewidth.
\newblock In {\em Proc. 12th IPEC}, volume~89 of {\em LIPIcs}, pages 4:1--4:12,
  2017.
\newblock \href {http://dx.doi.org/10.4230/LIPIcs.IPEC.2017.4}
  {\path{doi:10.4230/LIPIcs.IPEC.2017.4}}.

\bibitem{Berge}
C.~Berge.
\newblock Sur le couplage maximum d'un graphe.
\newblock {\em Comptes rendus hebdomadaires des séances de l'Académie des
  sciences}, 247:258 -- 259, 1958.

\bibitem{Binkele-RaibleFFLSV12}
Daniel Binkele-Raible, Henning Fernau, Fedor~V. Fomin, Daniel Lokshtanov, Saket
  Saurabh, and Yngve Villanger.
\newblock Kernel(s) for problems with no kernel: {On} out-trees with many
  leaves.
\newblock {\em ACM Trans. Algorithms}, 8(4):38, 2012.
\newblock \href {http://dx.doi.org/10.1145/2344422.2344428}
  {\path{doi:10.1145/2344422.2344428}}.

\bibitem{Bodlaender98}
Hans~L. Bodlaender.
\newblock A partial $k$-arboretum of graphs with bounded treewidth.
\newblock {\em Theor. Comput. Sci.}, 209(1-2):1--45, 1998.
\newblock \href {http://dx.doi.org/10.1016/S0304-3975(97)00228-4}
  {\path{doi:10.1016/S0304-3975(97)00228-4}}.

\bibitem{Bodlaender09}
Hans~L. Bodlaender.
\newblock Kernelization: New upper and lower bound techniques.
\newblock In {\em Proc. 4th {IWPEC}}, pages 17--37, 2009.
\newblock \href {http://dx.doi.org/10.1007/978-3-642-11269-0_2}
  {\path{doi:10.1007/978-3-642-11269-0_2}}.

\bibitem{BodlaenderD10}
Hans~L. Bodlaender and Thomas~C. van Dijk.
\newblock A cubic kernel for feedback vertex set and loop cutset.
\newblock {\em Theory Comput. Syst.}, 46(3):566--597, 2010.

\bibitem{BougeretS17}
Marin Bougeret and Ignasi Sau.
\newblock How much does a treedepth modulator help to obtain polynomial kernels
  beyond sparse graphs?
\newblock In {\em Proc. 12th IPEC}, volume~89 of {\em LIPIcs}, pages
  10:1--10:13, 2017.
\newblock \href {http://dx.doi.org/10.4230/LIPIcs.IPEC.2017.10}
  {\path{doi:10.4230/LIPIcs.IPEC.2017.10}}.

\bibitem{CyganLPPS14}
Marek Cygan, Daniel Lokshtanov, Marcin Pilipczuk, Michal Pilipczuk, and Saket
  Saurabh.
\newblock On the hardness of losing width.
\newblock {\em Theory Comput. Syst.}, 54(1):73--82, 2014.
\newblock \href {http://dx.doi.org/10.1007/s00224-013-9480-1}
  {\path{doi:10.1007/s00224-013-9480-1}}.

\bibitem{DellM14}
Holger Dell and Dieter van Melkebeek.
\newblock Satisfiability allows no nontrivial sparsification unless the
  polynomial-time hierarchy collapses.
\newblock {\em J. {ACM}}, 61(4):23:1--23:27, 2014.
\newblock \href {http://dx.doi.org/10.1145/2629620}
  {\path{doi:10.1145/2629620}}.

\bibitem{Fernau16}
Henning Fernau.
\newblock Kernelization, {T}uring kernels.
\newblock In {\em Encyclopedia of Algorithms}, pages 1043--1045. Springer,
  2016.
\newblock \href {http://dx.doi.org/10.1007/978-1-4939-2864-4_528}
  {\path{doi:10.1007/978-1-4939-2864-4_528}}.

\bibitem{FominJP14}
Fedor~V. Fomin, Bart M.~P. Jansen, and Michal Pilipczuk.
\newblock Preprocessing subgraph and minor problems: When does a small vertex
  cover help?
\newblock {\em J. Comput. Syst. Sci.}, 80(2):468--495, 2014.
\newblock \href {http://dx.doi.org/10.1016/j.jcss.2013.09.004}
  {\path{doi:10.1016/j.jcss.2013.09.004}}.

\bibitem{FominLMPS11}
Fedor~V. Fomin, Daniel Lokshtanov, Neeldhara Misra, Geevarghese Philip, and
  Saket Saurabh.
\newblock Hitting forbidden minors: Approximation and kernelization.
\newblock In {\em Proc. 28th STACS}, pages 189--200, 2011.
\newblock \href {http://dx.doi.org/10.4230/LIPIcs.STACS.2011.189}
  {\path{doi:10.4230/LIPIcs.STACS.2011.189}}.

\bibitem{FominLMPS16}
Fedor~V. Fomin, Daniel Lokshtanov, Neeldhara Misra, Geevarghese Philip, and
  Saket Saurabh.
\newblock Hitting forbidden minors: Approximation and kernelization.
\newblock {\em {SIAM} J. Discrete Math.}, 30(1):383--410, 2016.
\newblock \href {http://dx.doi.org/10.1137/140997889}
  {\path{doi:10.1137/140997889}}.

\bibitem{FominLMS12}
Fedor~V. Fomin, Daniel Lokshtanov, Neeldhara Misra, and Saket Saurabh.
\newblock Planar $\mathcal{F}$-{D}eletion: Approximation, kernelization and
  optimal {FPT} algorithms.
\newblock In {\em Proc. 53rd FOCS}, pages 470--479, 2012.
\newblock \href {http://dx.doi.org/10.1109/FOCS.2012.62}
  {\path{doi:10.1109/FOCS.2012.62}}.

\bibitem{FominLSZ19}
Fedor~V. Fomin, Daniel Lokshtanov, Saket Saurabh, and Meirav Zehavi.
\newblock {\em Kernelization: Theory of Parameterized Preprocessing}.
\newblock Cambridge University Press, 2019.
\newblock \href {http://dx.doi.org/10.1017/9781107415157}
  {\path{doi:10.1017/9781107415157}}.

\bibitem{FortnowS11}
Lance Fortnow and Rahul Santhanam.
\newblock Infeasibility of instance compression and succinct {PCP}s for {NP}.
\newblock {\em J. Comput. Syst. Sci.}, 77(1):91--106, 2011.
\newblock \href {http://dx.doi.org/10.1016/j.jcss.2010.06.007}
  {\path{doi:10.1016/j.jcss.2010.06.007}}.

\bibitem{GiannopoulouJLS15}
Archontia~C. Giannopoulou, Bart M.~P. Jansen, Daniel Lokshtanov, and Saket
  Saurabh.
\newblock Uniform kernelization complexity of hitting forbidden minors.
\newblock In {\em Proc. 42nd ICALP}, 2015.
\newblock In press.
\newblock \href {http://arxiv.org/abs/1502.03965} {\path{arXiv:1502.03965}}.

\bibitem{GiannopoulouJLS17}
Archontia~C. Giannopoulou, Bart M.~P. Jansen, Daniel Lokshtanov, and Saket
  Saurabh.
\newblock Uniform kernelization complexity of hitting forbidden minors.
\newblock {\em {ACM} Trans. Algorithms}, 13(3):35:1--35:35, 2017.
\newblock \href {http://dx.doi.org/10.1145/3029051}
  {\path{doi:10.1145/3029051}}.

\bibitem{GuoHN04}
Jiong Guo, Falk H{\"{u}}ffner, and Rolf Niedermeier.
\newblock A structural view on parameterizing problems: Distance from
  triviality.
\newblock In {\em Proc. 1st IWPEC}, pages 162--173, 2004.
\newblock \href {http://dx.doi.org/10.1007/978-3-540-28639-4_15}
  {\path{doi:10.1007/978-3-540-28639-4_15}}.

\bibitem{HermelinKSWW15}
Danny Hermelin, Stefan Kratsch, Karolina Soltys, Magnus Wahlstr{\"{o}}m, and
  Xi~Wu.
\newblock A completeness theory for polynomial ({T}uring) kernelization.
\newblock {\em Algorithmica}, 71(3):702--730, 2015.
\newblock \href {http://dx.doi.org/10.1007/s00453-014-9910-8}
  {\path{doi:10.1007/s00453-014-9910-8}}.

\bibitem{Iwata17}
Yoichi Iwata.
\newblock Linear-time kernelization for feedback vertex set.
\newblock In {\em Proc. 44th ICALP}, volume~80 of {\em LIPIcs}, pages
  68:1--68:14, 2017.
\newblock \href {http://dx.doi.org/10.4230/LIPIcs.ICALP.2017.68}
  {\path{doi:10.4230/LIPIcs.ICALP.2017.68}}.

\bibitem{Jansen17}
Bart M.~P. Jansen.
\newblock Turing kernelization for finding long paths and cycles in restricted
  graph classes.
\newblock {\em J. Comput. Syst. Sci.}, 85:18--37, 2017.
\newblock \href {http://dx.doi.org/10.1016/j.jcss.2016.10.008}
  {\path{doi:10.1016/j.jcss.2016.10.008}}.

\bibitem{JansenB13}
Bart M.~P. Jansen and Hans~L. Bodlaender.
\newblock Vertex cover kernelization revisited - upper and lower bounds for a
  refined parameter.
\newblock {\em Theory Comput. Syst.}, 53(2):263--299, 2013.
\newblock \href {http://dx.doi.org/10.1007/s00224-012-9393-4}
  {\path{doi:10.1007/s00224-012-9393-4}}.

\bibitem{JansenK13}
Bart M.~P. Jansen and Stefan Kratsch.
\newblock Data reduction for graph coloring problems.
\newblock {\em Inf. Comput.}, 231:70--88, 2013.
\newblock \href {http://dx.doi.org/10.1016/j.ic.2013.08.005}
  {\path{doi:10.1016/j.ic.2013.08.005}}.

\bibitem{JansenP18}
Bart M.~P. Jansen and Astrid Pieterse.
\newblock Polynomial kernels for hitting forbidden minors under structural
  parameterizations.
\newblock In {\em Proc. 26th ESA}, volume 112 of {\em LIPIcs}, pages
  48:1--48:15, 2018.
\newblock \href {http://dx.doi.org/10.4230/LIPIcs.ESA.2018.48}
  {\path{doi:10.4230/LIPIcs.ESA.2018.48}}.

\bibitem{JansenPW17}
Bart M.~P. Jansen, Marcin Pilipczuk, and Marcin Wrochna.
\newblock Turing kernelization for finding long paths in graphs excluding a
  topological minor.
\newblock In {\em Proc. 12th IPEC}, volume~89 of {\em LIPIcs}, pages
  23:1--23:13, 2017.
\newblock \href {http://dx.doi.org/10.4230/LIPIcs.IPEC.2017.23}
  {\path{doi:10.4230/LIPIcs.IPEC.2017.23}}.

\bibitem{KratschW14}
Stefan Kratsch and Magnus Wahlstr{\"{o}}m.
\newblock Compression via matroids: {A} randomized polynomial kernel for odd
  cycle transversal.
\newblock {\em {ACM} Trans. Algorithms}, 10(4):20:1--20:15, 2014.
\newblock \href {http://dx.doi.org/10.1145/2635810}
  {\path{doi:10.1145/2635810}}.

\bibitem{LewisY80}
John~M. Lewis and Mihalis Yannakakis.
\newblock The node-deletion problem for hereditary properties is {NP}-complete.
\newblock {\em J. Comput. Syst. Sci.}, 20(2):219--230, 1980.

\bibitem{Lokshtanov09}
Daniel Lokshtanov.
\newblock {\em New Methods in Parameterized Algorithms and Complexity}.
\newblock PhD thesis, University of Bergen, Norway, 2009.

\bibitem{LokshtanovMS12}
Daniel Lokshtanov, Neeldhara Misra, and Saket Saurabh.
\newblock Kernelization - {P}reprocessing with a guarantee.
\newblock In {\em The Multivariate Algorithmic Revolution and Beyond}, pages
  129--161, 2012.
\newblock \href {http://dx.doi.org/10.1007/978-3-642-30891-8_10}
  {\path{doi:10.1007/978-3-642-30891-8_10}}.

\bibitem{Niedermeier10}
Rolf Niedermeier.
\newblock Reflections on multivariate algorithmics and problem
  parameterization.
\newblock In {\em Proc. 27th STACS}, pages 17--32, 2010.
\newblock \href {http://dx.doi.org/10.4230/LIPIcs.STACS.2010.2495}
  {\path{doi:10.4230/LIPIcs.STACS.2010.2495}}.

\bibitem{RobertsonS95b}
Neil Robertson and Paul~D. Seymour.
\newblock Graph minors. {XIII.} {T}he disjoint paths problem.
\newblock {\em J. Comb. Theory, Ser. B}, 63(1):65--110, 1995.
\newblock \href {http://dx.doi.org/10.1006/jctb.1995.1006}
  {\path{doi:10.1006/jctb.1995.1006}}.

\bibitem{RobertsonS04}
Neil Robertson and Paul~D. Seymour.
\newblock Graph minors. {XX}. {W}agner's conjecture.
\newblock {\em J. Comb. Theory, Ser. B}, 92(2):325--357, 2004.
\newblock \href {http://dx.doi.org/10.1016/j.jctb.2004.08.001}
  {\path{doi:10.1016/j.jctb.2004.08.001}}.

\bibitem{Schrijver03}
Alexander Schrijver.
\newblock {\em Combinatorial Optimization. Polyhedra and Efficiency}.
\newblock Springer, Berlin, 2003.

\bibitem{Thomasse10}
St{\'e}phan Thomass{\'e}.
\newblock A $4k^2$ kernel for feedback vertex set.
\newblock {\em ACM Trans. Algorithms}, 6(2), 2010.
\newblock \href {http://dx.doi.org/10.1145/1721837.1721848}
  {\path{doi:10.1145/1721837.1721848}}.

\bibitem{UhlmannW13}
Johannes Uhlmann and Mathias Weller.
\newblock Two-layer planarization parameterized by feedback edge set.
\newblock {\em Theor. Comput. Sci.}, 494:99--111, 2013.
\newblock \href {http://dx.doi.org/10.1016/j.tcs.2013.01.029}
  {\path{doi:10.1016/j.tcs.2013.01.029}}.

\bibitem{Weller2013}
Mathias Weller.
\newblock {\em Aspects of Preprocessing Applied to Combinatorial Graph
  Problems}.
\newblock PhD thesis, Technische Universit\"at Berlin, 2013.

\end{thebibliography}
